\documentclass{article}

\usepackage{PRIMEarxiv}

\usepackage[utf8]{inputenc} 
\usepackage[T1]{fontenc}    
\usepackage{hyperref}       
\usepackage{url}            
\usepackage{booktabs}       
\usepackage{amsfonts}       
\usepackage{nicefrac}       
\usepackage{microtype}      

\usepackage{fancyhdr}       
\usepackage{graphicx}       
\usepackage[flushleft]{threeparttable}
\usepackage{amsmath,amssymb,amsthm,color,epsfig}
\usepackage[labelfont=bf]{caption}
\usepackage[numbers]{natbib}

\usepackage{thm-restate}
\usepackage{bbm, dsfont}
\usepackage[linesnumbered,ruled,vlined]{algorithm2e}
\theoremstyle{plain}
\newcommand{\argmin}{\mathop{\rm argmin}}

\usepackage{comment}
\newcommand{\E}{{\mathbb{E}}}
\newcommand{\Var}{\mathrm{Var}}

\newtheorem{lemma}{Lemma}

\theoremstyle{definition}
\newtheorem{definition}{Definition}

\title{Marginally constrained nonparametric Bayesian inference through Gaussian processes
}

\author{
  Bingjing Tang \\
  Department of Statistics \\
  Purdue University \\
  \texttt{tang272@purdue.edu} \\
   \And
  Vinayak Rao \\
  Department of Statistics \\
  Purdue University \\
  \texttt{varao@purdue.edu} \\
}

\usepackage{cleveref}
\begin{document}
\maketitle

\begin{abstract}
Nonparametric Bayesian models are used routinely as flexible and powerful models of complex data. Many times, a statistician may have additional informative beliefs about data distribution of interest, e.g., its mean or subset components, that is not part of, or even compatible with, the nonparametric prior. An important challenge is then to incorporate this partial prior belief into  nonparametric Bayesian models. 
In this paper, we are motivated by settings where practitioners have 
additional distributional information about a subset of
{the coordinates of the observations} being modeled.
Our approach links this problem to that of conditional density modeling.
Our main idea is a novel constrained Bayesian model, based on a perturbation of a parametric distribution with a transformed Gaussian process prior on the perturbation function.
We also develop a corresponding posterior sampling method based on data augmentation. We illustrate the 
{efficacy of our proposed constrained nonparametric Bayesian model in a variety of  real-world scenarios including modeling environmental and earthquake data.}
\end{abstract}


\section{Introduction\label{sec:1}}
Nonparametric Bayesian methods like the Dirichlet process~\citep{ferguson} and the Gaussian process~\citep{rasmussen} allow practitioners to specify flexible  priors over infinite-dimensional objects like functions and probability densities.
These have seen wide success in applied disciplines like biostatistics~\citep{dunson}, document modeling~\citep{teh} and image  modeling~\citep{sudderth}, with an accompanying rich literature on theoretical properties and computational strategies.
The flexibility of these models 
however often comes at the cost of interpretibility, and they  can be quite challenging to elicit from applied scientists. 
Even more problematic is that such flexible priors imply statistical properties of the objects of interest that are incompatible with expert knowledge. This expert knowledge is often quantified by probability distributions over functionals of the infinite-dimensional objects of interest. 
%
The challenge is then to specify nonparametric Bayesian priors subject to constraints on the distribution of such functionals.
This problem can be quite general~\citep{kessler}, and in this work, we focus on a specific setting.
Consider observations lying in a $d$-dimensional Euclidean space $\Re^d$, with the observations drawn i.i.d.\ from an unknown probability density. 
We model this density, an infinite-dimensional object, with a nonparametric Bayesian prior, but now wish to constrain the marginal distribution of a subset of the coordinates of the observations. 

{\Cref{fig:marginal constraint graphic interpretation} presents a graphical illustration of this problem. The two-dimensional contour plot represents a complex probability density $p(X_1,X_2)$, drawn from a nonparametric prior. The one-dimensional densities to the top and right, plotted with broken lines, show the corresponding marginal densities of $X_1$ and $X_2$ respectively.
If the first component $X_1$ is known to follow some known density (e.g. a Gaussian, shown with the continuous curve), then it is important to modify the prior to satisfy this, while still remaining flexible about the rest of the density.}
\begin{figure}
  \centerline{
  \includegraphics[width=0.8\textwidth]{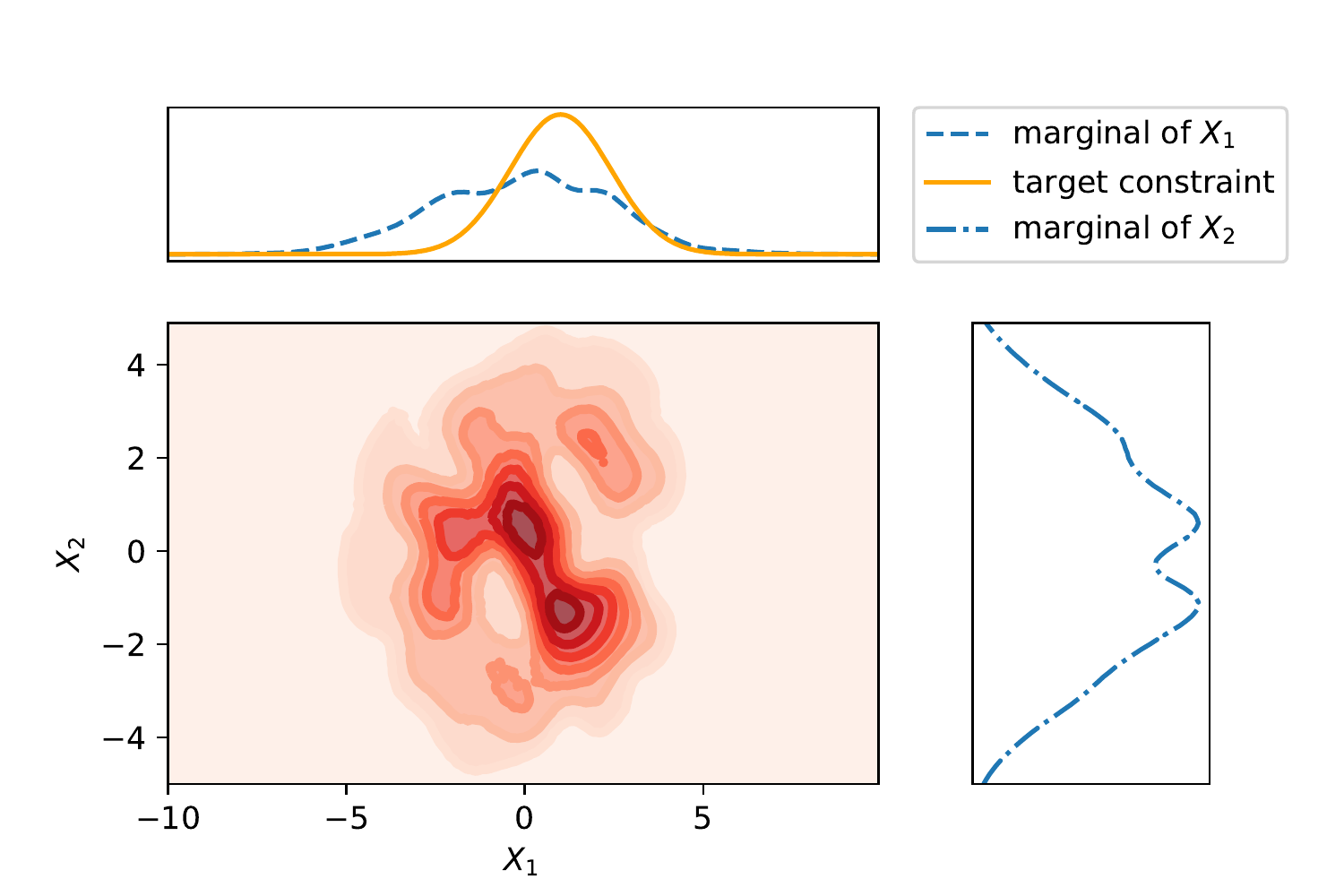}

  }
  \caption{Contour plot of a two-dimensional density $p(X_1, X_2)$ drawn from a nonparametric prior. Right is the corresponding marginal density of $X_2$, while top is the corresponding marginal density of $X_1$, along with a target constraint $p(X_1)$. \label{fig:marginal constraint graphic interpretation}}
\end{figure}

As a motivating example, following~\citet{schifeling}, we consider the 2012 American Community Survey Public Use Microdata Sample of North Carolina survey data (ACS PUM).  
This dataset, obtained from the United States Census Bureau's website (\url{https://data.census.gov/mdat/#/}), comprises features like gender, age, and educational attainment. 
The marginal distribution of age in the population may already be known empirically from external census data sources and, due to sampling effects, may differ from the empirical distribution of age in the survey data. 
Modeling this dataset with an off-the-shelf nonparametric model will imply a marginal distribution over age that also differs from the prior knowledge.
Incorporating the marginal distribution of age into the off-the-shelf prior can result in more accurate inferences, as predictive datasets from the model will satisfy 
the marginal distribution of age group in the population while {preserving} the dependence structure in the original multivariate data. 

More broadly, such an approach is also useful in a variety of modern tasks in statistics and machine learning.
A topic of increasing interest in machine learning is the problem of distribution shift, when the target distribution does not align with the underlying distribution of the training data. 
\citet{dai2022marginal} proposed an approach for such a setting, but for discrete set modeling 
(see~\cref{sec:related}).
Similarly, concerns about fairness and privacy might make it preferable to introduce simplifying constraints on certain variables, rather than modeling the observed data as accurately as possible.
\section{Related work} \label{sec:related}
An important step towards incorporating marginal constraints into nonparametric models, motivated by the ACS PUM survey dataset mentioned earlier, is the work of~\citet{schifeling}. 
Here, the authors used a Dirichlet process mixture of products of multinomials to model a dataset of discrete values.
To enforce marginal constraints on a subset of the coordinates, the authors proposed a {\em hypothetical records augmentation method}, which essentially amounts simulating values of these coordinates from the specified marginal distribution.
These simulated values serve as an auxiliary dataset, which when combined with the original dataset, guide results towards respecting the marginal constraint.
This data-augmentation approach, while conceptually simple, has a number of limitations. 
First, it only approximately enforces the marginal constraint, with the constraint enforced exactly only in the limit as the number of augmenting datapoints tends to infinity.
The authors do not provide any clear guidelines about choosing the size of the augmenting dataset.
Since the unconstrained coordinates are missing on the augmenting dataset, posterior inference involves additional complexity, and the scheme in~\citet{schifeling} is tailored only to categorial data.
Finally, to simulate the augmenting dataset, the marginal density must be known exactly. 
In many settings, one only wishes to constrain the parametric family the marginal distribution belongs to, with its parameters themselves learned from the data.
For instance, prior knowledge might suggest that interarrival times in a queue follows an exponential distribution, with the actual arrival rate unknown.
In such a setting, 
one cannot directly apply the methodology from~\citet{schifeling}.

Our proposed approach links the problem of enforcing marginal constraints to the problem of nonparametric conditional density modeling~\citep{dunson2008kernel,chung2009nonparametric,pati2013posterior,ghosh2010bayesian,tokdar2011dimension}.
Rather than indirectly induce a marginal constraint on a subset of variables using an auxiliary dataset, we propose to directly model that subset using the specified marginal distribution, placing a prior on any unknown parameters.
We then place a nonparametric prior on the conditional density of the remaining variables, with the resulting joint distribution 1) satisfying the marginal constraint {\em by construction}, and 2) inheriting the large support and flexibility of the nonparamtric prior.
Our contribution in this paper can thus also be viewed as a novel model for nonparametric conditional density modeling, with an associated novel MCMC sampling algorithm.
As we outline below, there already exist a number of approaches to nonparametric conditional density modeling in the literature. 
We note that when the marginal constraint is known exactly, our problem reduces to a conditional density modeling problem, and any of these existing methods can be used.
When only the parametric form of the marginal constraint is known, then our model affords a little more flexibility by allowing the conditional distribution to vary with the marginal.
As an additional benefit, by building on the work of~\citet{adams}, our model has an associated {\em exact} MCMC sampling with no asymptotic bias, something that is lacking for most existing conditional density models.

There is a rich and growing literature on nonparametric conditional density modeling, with a large number of approaches based on {predictor-dependent} stick-breaking process priors. {To estimate the conditional density of a response variable $Y$, these use the mixture specification 
$p(Y|X)=\int g(Y|X, \psi) dG_X(\psi)$, with known parametric densities $g(Y|X,\psi)$.
The random probability measure $G_X$ is a predictor-dependent mixture
distribution taking the form 
$$G_X=\sum_{h=1}^\infty \pi_h(X)\delta_{\psi_h},\; \psi_h\sim G_0,$$ where $G_0$ is a basis measure and $\pi_h (X)$ is a predictor-dependent probability weight constructed from a predictor-dependent stick-breaking process$$\pi_h(X)=V_h(X)\prod_{s<h}\{1-V_s(X)\}$$ with $V_h$
confined to the unit interval for all $X$. Different construction choices of $V_h$ have been proposed, among others works,  in~\citet{dunson2008kernel} and \citet{chung2009nonparametric}.}
While well understood theoretically\citep{pati2013posterior}, computation with these models typically requires truncating the number of mixture components to some finite number, or implementing involved slice-sampling algorithms.
An alternative approach based on logistic Gaussian process is introduced in~\citet{ghosh2010bayesian,tokdar2011dimension}.
The authors extend logistic Gaussian process priors originally  introduced and
studied by~\citet{lenk1988logistic,lenk1991towards,leonard1978density} to model a nonparametric conditional density. Again however, exact posterior computation with logistic Gaussian process priors is a difficult problem as indicated in~\citet{ghosh2010bayesian, tokdar2007towards}. To avoid this issue, we extend the sigmoid Gaussian process prior introduced in~\citet{adams} for modeling a nonparametric conditional density, and show how one can carry out exact posterior inference when this class of models is extended to modeling conditional distributions.  

In~\citet{kessler}, the authors considered a version of the broader problem stated at the start of this paper: given a prior $\Pi_0(\theta)$ on some joint space $\Theta$, and a marginal distribution $p_1$ on some functional $f$ of $\theta$, what is the probability measure `closest' to $\Pi_0$ that satisfies the marginal constraint.
Using the Kullback-Leibler divergence as the measure of closeness, and writing 
$\mathcal{H}$ for the set of probability measures on $\Theta$ with $f-$marginal distribution $p_1$, the authors arrived at the following solution: 
\begin{align}
\Pi_1(\theta) &:= \argmin_{\Pi \in \mathcal{H}}\text{KL}(\Pi(\theta)|\Pi_0(\theta)) =\Pi_0(\theta|f)p_1(f). \label{eq:kessler}
\end{align}
 The form of this solution allowed the authors to carry out posterior sampling via a small modification to an MCMC algorithm for the unconstrained prior $\Pi_0(\theta)$.
 While conceptually simple, this correction involves computing an intractable marginal, and in practice, this algorithm must be run as an approximate one, requiring a kernel density estimate of this marginal probability in order to calculate the acceptance probability.
We can attempt to cast our problem as a special instance of this general scheme, where $\Theta$ is the space of densities on $\Re^d$ and $f$ projects $\theta$ onto a density on a subset of these coordinates.
Note though that with a known marginal constraint, the projected variable is distributed as a Dirac delta function.
As a consequence of this hard constraint, 
the earlier MCMC scheme will make proposals attempting to hit a submanifold of measure 0, resulting in an acceptance probability of $0$.
Our proposed solution nevertheless builds on the form given in~\cref{eq:kessler}, and links the problem of marginal constraints to that of nonparametric conditional density modeling.

{As mentioned at the end of~\cref{sec:related}, \citet{dai2022marginal} consider a similar problem is the machine learning setting of marginal distribution shift. Their approach is restricted to discrete sets, and different from our problem, the authors try to learn a generative model to approximate the distribution of a random discrete {\em set}, 
given information about element marginals, i.e., the occurrence frequency of particular elements. 
Their goal is to efficiently adapt a previously learned generative model without constraints to respect the element marginals, 
without having to retrain the entire generative model from scratch.}

\section{Problem statement and proposed approach}

In the following, we denote random variables with uppercase letters and their values with lowercase letters.
We seek to model a dataset $\mathbf{X} = (X^1,\dotsc,X^n)$ comprising $n$ observations, each of dimension $d$, with the $i$th observation written as $X^i = (X^i_1,\dotsc, X^i_d)$.
As this notation indicates, we use superscripts to index individual observations in the dataset, and subscripts to index covariates.
For an increasing, ordered subset $A \subset \{1,\dotsc,d\}$, we will use $X_A$ to refer the subvector $(X_j)_{j \in A}$ of $X$,  and write $X_{A^c}=(X_j)_{j \notin A}$ for the complement of the components $X_A$.
We model the observations as independent and identical draws from some unknown probability density $p$, subject to the following constraint: the subset $X_A$ follows a known marginal distribution $p_A$.
This marginal constraint on $X_A$, obtained from some external source, can take the form of a completely specified distribution (e.g.\ $p_A$ might be the standard normal distribution $N(0,1)$).
More generally, domain knowledge about the distribution of $X_A$ might take the form of some parametric distribution with the parameter value unknown, for example, $X_A$ might be marginally distributed as a Gaussian with unknown mean and variance.
We refer to the former as a {\em distribution constraint}, and the latter as a {\em family constraint}.
The latter arises in fields like queuing theory or genetics, where structural knowledge like memorylessness or independence can result in distributions like the Poisson, Gaussian or exponential.
In the first real example presented in the~\cref{sec:3}, the modeler knows that  concentrations of pollutants follow a lognormal distribution, with the mean and variance unknown. 
Note that the hypothetical records approach of~\citet{schifeling} no longer works, since know we no longer have a specific distribution to generate auxiliary samples.
Taking a nonparametric Bayesian approach, the problem now is to place a nonparametric Bayesian prior on the unknown density $p$, while ensuring that the induced marginal distribution of $X_A$ is consistent with this side knowledge.  
For the family constraint, one must also place a prior on the unknown parameters.

\subsection{ Proposed nonparametric Bayesian Model}
\label{section:newmodel}
Our modeling approach, which can also be viewed as a contribution to the literature on nonparametric conditional density modeling, proceeds as follows.
Inspired by~\cref{eq:kessler}, we factor the joint probability as $p(X_A,X_{A^c}) = p_A(X_A)p(X_{A^c}|X_A)$. 
We then directly model the subset $X_A$ of variables according to their specified marginal distribution, 
and then place a flexible nonparametric prior on the conditional density of the remaining variables given a realization of $X_A$.
Specifically, with $p_A(\;\cdot\;|\;\phi)$ the known marginal constraint (parametrized by $\phi$), we model $X_A$ as
\begin{align}
  X_A \sim p_A(\;\cdot\;|\;\phi), \qquad \phi \sim p_\phi(\cdot).
\end{align}
When the prior $p_\phi$ on $\phi$ is a Dirac delta (that is, $\phi$ is known), we are in the distributional constraint setting.
For the family constraint setting, the prior $p_\phi$ on $\phi$ might either reflect domain knowledge about $\phi$, or can be a weakly informative or uninformative distribution.

We now have to specify a prior on the conditional distribution of $X_{A^c}$ given $X_A$.
As outlined earlier, there exist a number of approaches in the literature to do this, though posterior inference for these typically involve discretization or truncation approximations.
After specifying our prior below, we will show in the next section how it is possible to carry out {\em exact} MCMC inference for this model.

We next introduce some more notation. Let $\pi_0(X_{A^c}|X_A, \theta, \phi)$ be a simple conditional distribution that is strictly positive, continuous, and parameterized by $\theta$ and $\phi$. 
This will serve as a {\em centering distribution}.
We will model the true conditional distribution of $X_{A^c}$ given $X_A$ as a perturbation of this centering distribution, with a transformed Gaussian process (GP)~\citep{rasmussen} prior on the perturbation function.
Specifically, let $\sigma(x) = 1/(1+\exp(-x))$ be the sigmoid function, and let $\lambda(\cdot)$ be a realization of a Gaussian process with mean function $\mu(\cdot)$ and covariance kernel $ k(\cdot,\cdot)$. 
Then, we set the conditional distribution $p(X_{A^c}|X_A,\theta,\phi) \propto \pi_0(X_{A^c}|X_A, \theta,\phi)\sigma(\lambda(X_A,X_{A^c})) $. 
Recall that $\lambda$ being a Gaussian process implies that for every finite subset of random vectors $z_1, \dots, z_m \in \Re^d$,  the corresponding vector $\lambda(z_1),\dotsc,\lambda(z_m)$ follows a Gaussian distribution as below, see~\citet{rasmussen} for more details:
 \begin{align*}
    \begin{pmatrix}
\lambda(z_1) \\
\vdots\\
\lambda(z_m) \\
\end{pmatrix}\sim \mathcal{N}\left(\begin{pmatrix}
\mu(z_1) \\
\vdots\\
\mu(z_m) \\
\end{pmatrix},\begin{pmatrix}
k(z_1, z_1) & \dots & k(z_1, z_m) \\
&\ddots&\\
k(z_m, z_1) &\dots & k(z_m, z_m) \\
\end{pmatrix}\right). 
 \end{align*}
The Gaussian property above shows that sample paths of a GP can take negative values, and the sigmoid transformation serves as a {\em link function} to keep $p(X_{A^c}|X_A,\theta,\phi)$ positive.
There are a number of other choices for the link function in the literature, the  most popular being the exponential function in~{\citet{tokdar2007towards}}. 
However, we show that along these lines of~\citet{adams}, the use of the sigmoid function allows us to sample exactly from this model.
Following~\citet{rao}, we can carry out exact MCMC inference without any approximation error.
First, we write down the overall model below:
 \begin{align}
 \theta &\sim p_\theta(\cdot), \quad \phi \sim p_\phi(\cdot), \quad \lambda(\cdot) \sim \mathcal{GP}(\mu(\cdot), k(\cdot, \cdot)) \label{eq:eqn9}\\
 X_A &\sim p_A(\;\cdot\;|\;\phi)\label{eq:eqn10}\\
 X_{A^c}|X_A & \sim \frac{\pi_0(X_{A^c}|X_A,\theta, \phi)\sigma(\lambda(X_A, X_{A^c}))}{\int \pi_0(X_{A^c}|X_A,\theta,\phi)\sigma(\lambda(X_A,X_{A^c}))     dX_{A^c}}.
 \label{eq:eqn11}
 \end{align}
The Gaussian process is a well studied nonparametric prior on functions, and can be shown to possess desirable large support properties~\citep{choudhuri2007nonparametric}.
By modeling the modulation function with a GP prior, one can expect the resulting conditional density also to inherit similar large support properties. We show this below.

\subsection{Large support and consistency properties}
Assume  that $\mathcal{X}_A\times \mathcal{X}_{A^c}$ is a compact subset of  {$\Re^d$}.
For simplicity, we assume that the parameters $\theta$ and $\phi$ are fixed and known, and drop them from all notation, writing the centering distribution and the marginal constraining distribution simply as $\pi_0(\cdot|\cdot)$ and $p_A(\cdot)$. 
Write $\mathcal{F}$ for the space of all joint densities on $\mathcal{X}_A\times \mathcal{X}_{A^c}$ with {their corresponding} marginal densities on $\mathcal{X}_A$ equal to $p_A(\cdot)$, and conditional densities jointly continuous on $\mathcal{X}_A\times \mathcal{X}_{A^c}$. 
The Gaussian process prior on $\lambda$ induces a prior $\Pi$ on $\mathcal{F}$ through the map {from $\lambda$ to $f_\lambda$}
\begin{align}
&\lambda\rightarrow f_\lambda(X_A, X_{A^c})=p_A(X_A)\cdot\cfrac{\pi_0(X_{A^c}| X_A)\sigma(\lambda(X_A,X_{A^c}))}{\int_{\mathcal{X}_{A^c}} \pi_0(X_{A^c}| X_A)\sigma(\lambda(X_A,X_{A^c}))dX_{A^c}}.\label{eq:map}
\end{align}
In~\cref{thm:largesupport}, we prove a `large-support' property of the conditional density in~\cref{eq:eqn11}, showing that any density function in $\mathcal{F}$ is in the KL support of the induced prior $\Pi$. 
First, we state a key intermediate result from~\citet{ghosal2006posterior} that we will use to prove the theorem:
\begin{restatable}{lemma}{xyz}{[Theorem 4 in \citet{ghosal2006posterior}]}
  Assume that $\lambda(X_A,X_{A^c})$ on the compact index set $\mathcal{X}_A \times \mathcal{X}_{A^c}$ is a Gaussian process (GP) with continuous sample paths. 
Assume that the GP mean function and a function  $\lambda_0(X_A, X_{A^c})$  on $\mathcal{X}_A \times \mathcal{X}_{A^c}$ belong to the RKHS of the covariance kernel of the Gaussian process. Then
 \begin{align*}
     P(\lambda:\sup_{\{X_A\in \mathcal{X}_A,\,    X_{A^c} \in \mathcal{X}_{A^c}\}}|\lambda(X_A, X_{A^c})-\lambda_0(X_A,X_{A^c})|<\delta)>0\;\; \forall \delta>0.
 \end{align*}\label{lm:gaussian process property1}
\end{restatable}
We can then state our result, which we prove in Appendix A~\ref{appendix A}.
\begin{restatable}{theorem}{abc}
Suppose the Gaussian process 
on the compact space $\mathcal{X}_A\times\mathcal{X}_{A^c}$  satisfies the assumptions in~\cref{lm:gaussian process property1}. 
Assume that its mean function is continuous and that the RKHS associated with its covariance kernel equals the set of all continuous functions on $\mathcal{X}_A\times\mathcal{X}_{A^c}$. 
Also assume $\pi_0(X_{A^c}|X_A)$ and $p_A(X_A)$ in the map of equation~\ref{eq:map} are strictly positive on $\mathcal{X}_A\times\mathcal{X}_{A^c}$ and $\mathcal{X}_A$, and additionally $\pi_0(X_{A^c}|X_A)$ is continuous on $\mathcal{X}_A\times\mathcal{X}_{A^c}$. Then any  density function belonging to $\mathcal{F}$ is in the KL support of  $\Pi$.\label{thm:largesupport}
\end{restatable}
We remark that for many covariance kernels, the associated RKHS equals the set of all continuous functions, see theorems 4.3-4.5 of \citet{tokdar2007posterior}. 
If a kernel on a one-dimensional space can be written as $k(s,t)=\psi(s-t)$ for some nonzero, continuous density function $\psi$, then \citet{tokdar2007posterior} showed its RKHS is the set of continuous functions.
For spaces with dimension larger than 1, if the covariance kernel is the Kronecker product of one-dimensional kernels, with each one-dimensional kernel having the set of all continuous functions as its RKHS, then 
\citet{tokdar2007posterior} showed that the RKHS of the covariance kernel also equals the set of all continuous functions.  

We finish this section with a well-known result from~\citet{schwartz1965bayes}, showing that the large support property of the prior distribution from~\cref{thm:largesupport} translates to asymptotic consistency of the posterior distribution.
Specifically, assume the product of $p_A(X_A)$ and the true conditional density $q(X_{A^c}|X_A)$ lies in $\mathcal{F}$, then the posterior $\Pi(f_\lambda|(X_A^1,X_{A^c}^1),\dots,(X_A^n,X_{A^c}^n))$ will be weakly consistent at the joint density $p_A(X_A)\cdot q(X_{A^c}|X_A)$ as $n\rightarrow \infty$.
\begin{restatable}{theorem}{ef}
(\citet{schwartz1965bayes})
If the true density $f_0$ is in the Kullback-Leibler support of $\Pi$, then the posterior is weakly consistent at $f_0$. \label{thm:weak_cons}
\end{restatable}

\subsection{Exact prior simulation}
We next explain how to draw samples from this model. This step, useful in itself, is also key to our MCMC algorithm for posterior simulation.
Along the lines of~\citet{adams}, we exploit the fact that the logistic function $\sigma(x)$ satisfies $\sigma(x) \le 1$, so that 
$\pi_0(X_{A^c}|X_A; \theta,\phi) \sigma(\lambda(X_A, X_{A^c})) \le \pi_0(X_{A^c}|X_A; \theta,\phi)$. We note that this bounding property does not hold for most typical link functions used in the literature.
The bounding property allows us to generate samples from the nonparametric conditional through a simple rejection sampling scheme: first propose from $ \pi_0(X_{A^c}|X_A; \theta,\phi)$ and then accept or reject with probability $\sigma(\lambda(X_A,X_{A^c}))$.
The accepted samples form a realization from the probability density proportional to $\pi_0(X_{A^c}|X_A; \theta,\phi)\sigma(\lambda(X_A, X_{A^c}))$.
Crucially, generating a dataset of $n$ observations in such a fashion only requires evaluating the Gaussian process at a finite set of points: the locations of the data points as well as the rejected proposals that were produced along the way. 
Since the values of a GP on a finite set of {points} follow a multivariate {normal} distribution, these can easily be sequentially sampled from the corresponding conditional distributions. 
Importantly, this does not require integrating the transformed GP as in the denominator of~\cref{eq:eqn11}.
We describe this in detail in~\cref{algo:4}. 
 \begin{algorithm}
  \KwIn{the marginal density family $p_A(\cdot)$,  a tractable conditional density family $\pi_0(\cdot|\cdot)$, priors $p_\phi(\cdot)$ and $p_\theta(\cdot)$ on $\phi$ and $\theta$, the Gaussian process mean and covariance functions $\mu(\cdot)$ and $k(\cdot,\cdot)$, the sample size $n$.}
  \KwOut{Multivariate samples $(X_{A}^1, X_{A^c}^1), \dots, (X_{A}^n, X_{A^c}^n)$ with $X^i_A\sim p_A(\cdot)$.}
  Set $\lambda_{pre}$ as null and $S=\{1,2,\dots,n\}$.
  
  Sample
  $\phi\sim p_\phi(\cdot)$ and $\theta\sim p_\theta(\cdot)$.
  
  Sample $X_{A}^1,\dots, X_{A}^n \stackrel{i.i.d}{\sim} p_A(\cdot|\phi)$. 
  
 \SetAlgoLined\SetArgSty{}
 \Repeat{$S$ is null}
 {
 
 Sample $y^{i} \sim \pi_0(\cdot|X_{A}^i; \theta,\phi)$ $\forall i\in S$, and denote $T=\{(X_{A}^i, y^{i}),i \in S\}$;
 
 Sample $\lambda_T|\lambda_{pre}$ conditionally from the Gaussian process, and add $\lambda_T$ into $\lambda_{pre}$;
 
 Set $\alpha^i=1$ with probability $\sigma(\lambda(X_{A}^i, y^{i}))$ $\forall i \in S$;
 
  Set $X_{A^c}^i \gets y^{i}$   and delete index $i$ from $S$ if $\alpha^i=1 \,\forall i \in S$.
  }
  \caption{Algorithm to generate prior samples from the proposed model (Equations~\eqref{eq:eqn9} -~\eqref{eq:eqn11})}
  \label{algo:4}
\end{algorithm}

\subsubsection{Placing a prior distribution on { $\theta$ and $\phi$:} the parameters of the centering distribution  {$\pi_0(\;\cdot\;|\;\cdot\;;\theta, \phi)$}}\label{section:prior on conditional}
As the results in~\cref{thm:largesupport} and~\cref{thm:weak_cons} show, it is sufficient for the centering density $\pi_0(\;\cdot\;|\;\cdot\;;\theta, \phi)$ to be a strictly positive, {continuous} conditional distribution for our model to possess desirable asymptotic properties. 
Even so, a poor choice of $\pi_0(\;\cdot\;|\;\cdot\;;\theta, \phi)$ can impact the efficiency of~\cref{algo:3}. 
Specifically, if there exists a region of space that has high probability under the true density, but low probability under the centering distribution, then it will take a large number of rejected proposals before a sample is finally accepted.
While ultimately avoiding such issues requires a careful choice of the centering distribution, in our experiments, we found that the additional flexibility gained from allowing a location and/or scale parameter of the centering distribution to vary significantly improves performance.
Accordingly, we place {priors $p_\theta(\cdot)$ and $p_\phi(\cdot)$} on the parameters of $\pi_0(\;\cdot\;|\;\cdot\;;\theta, \phi)$.

 \section{Posterior inference}\label{section:posterior computation}
 Having specified the model, we now move to the problem of posterior computation: given a dataset $(X^1,\dotsc,X^n)$ of observations from an unknown density, if we assume the data-generating density lies in $\mathcal{F}$ and model it using our marginally constrained nonparametric prior, how do we characterize resulting the posterior distribution?
We take a Markov chain Monte Carlo (MCMC) approach, devising a Markov chain with this posterior distribution as its stationary distribution.
 The unknown latent variables of interest in the model specified above are the GP-distributed function $\lambda$, and any unknown parameters $\theta$ and $\phi$ of the centering distribution and the marginal constraining distribution respectively. 
 For simplicity, we ignore any unknown hyperparameters of the kernel of the Gaussian process; these can easily be simulated given realizations of $\lambda$.
The posterior distribution for $(\lambda,\phi,\theta)$ is
\begin{align}
  p(\lambda,\phi, \theta|X^1,\dots,X^n)
  &\propto p_\phi(\phi)\cdot p_\theta(\theta)\cdot \mathcal{GP}(\lambda)\cdot \prod_{i=1}^n p(X^i|\lambda,\phi,\theta)\nonumber\\
    &\propto p_\phi(\phi) \cdot p_\theta(\theta)\cdot \mathcal{GP}(\lambda)\cdot\prod_{i=1}^n p_A(X_A^i|\phi)\cdot\prod_{i=1}^n \frac{\pi_0(X_{A^c}^i|X_A^i,\theta, \phi)\sigma(\lambda(X_{A}^i, X_{A^c}^i))}{Z(\lambda, \theta, \phi, X_{A}^i)}
    \label{eq:eqn121}
\end{align}
where $\mathcal{GP}(\lambda)$ denotes the Gaussian process prior and $Z(\lambda, \theta,  \phi, X_{A}^i)=\int \pi_0(X_{A^c}^i|X_{A}^i,\theta, \phi)\sigma(\lambda(X_{A}^i, X_{A^c}^i)) dX_{A^c}^i$.

Note that
the dependence of the intractable denominator $Z(\lambda, \theta, \phi, X_{A}^i)$ on $\theta$, $\phi$ and $\lambda$ makes their posterior distribution an example of a {\em doubly intractable probability distribution}~\citep{murray2012mcmc}, so that standard MCMC methods cannot directly be used.
Specifically, for a Metropolis-Hastings algorithm that proposes new parameters $\theta^*$, $\phi^*$ and a new function $\lambda^*$, the acceptance probability involves the ratios $Z(\lambda, \theta, \phi, X_A^i)/Z(\lambda^*, \theta^*, \phi^*, X_A^i)$, something that is clearly impossible to calculate.
To solve this, we follow a data augmentation scheme based on an approach proposed in~\citet{rao}.

At a high level, our approach is to also instantiate the rejected proposals from $\pi_0$. Write $\mathcal{Y}^i=\{y^{i1}, \dots, y^{i\left|\mathcal{Y}^i\right|}\}$ as the set of rejected samples preceding the $i$th observation $(X^i_A,X^i_{A^c})$.
Instead of computing $p(\lambda,\phi, \theta|(X_{A}^1, X_{A^c}^1), \dots, (X_{A}^n, X_{A^c}^n))$, we simulate from $p(\lambda,\phi,\theta,\mathcal{Y}^1, \dots, \mathcal{Y}^n| (X_{A}^1, X_{A^c}^1),  \dots, (X_{A}^n, X_{A^c}^n))$.
Observe that the latter has the former as its marginal distribution, so that having produced samples from the latter, we can just discard the rejected samples $(\mathcal{Y}^1,\dotsc,\mathcal{Y}^n)$.
Importantly, we will see that given the rejected samples $(\mathcal{Y}^1,\dotsc,\mathcal{Y}^n)$, the variables $\lambda$, $\phi$ and $\theta$ can be updated using standard MCMC techniques.
Our overall approach to produce samples from $p(\lambda,\phi,\theta,\mathcal{Y}^1, \dots, \mathcal{Y}^n| (X_{A}^1, X_{A^c}^1),  \dots, (X_{A}^n, X_{A^c}^n))$ is a Gibbs sampling algorithm  described in ~\cref{algo:3}, involving three main steps:
\begin{description}
\item[Sample $\mathcal{Y}^1, \dots, \mathcal{Y}^n \,|\, \lambda,\phi, \theta, (X_{A}^1, X_{A^c}^1) \dots, (X_{A}^n, X_{A^c}^n)$:]
Note that under our model, the number of rejected samples preceding each observation is a random quantity following a geometric distribution, whose success probability equals the acceptance probability. 
Noting that the acceptance probability equals $\int \pi_0(y|X_A^i;\theta,\phi)\sigma(\lambda(y,X_A^i))dy$ and is intractable, we avoid evaluating this by directly simulating the number and values of the rejected samples preceding an observation.
To do so, we just simulate a new observation following~\cref{algo:4}, discard the accepted sample and keep the rejected samples.
The validity of this scheme was proved in~\citet{rao} who showed that crucially, these quantities are independent of the value of the accepted sample. 
\item[Sample $\lambda\, |\,\phi,\theta, \mathcal{Y}^1, \dots, \mathcal{Y}^n, (X_{A}^1, X_{A^c}^1) \dots, (X_{A}^n, X_{A^c}^n)$:]
This conditional distribution no longer involves any intractable normalizers:
    \begin{align} p(\lambda|\,\phi,\theta, \mathcal{Y}^1, \dots, \mathcal{Y}^n,  (X_{A}^1, X_{A^c}^1), \dots, (X_{A}^n, X_{A^c}^n)) \propto p(\lambda) \cdot \prod_{i=1}^n \left\{\sigma(\lambda(X_{A}^i, X_{A^c}^i))\prod_{j=1}^{|\mathcal{Y}^i|}\left[1-\sigma(\lambda(X_{A}^i, y^{ij}))\right]\right\}.\label{eq:lambda}
        \end{align}
Effectively, $\sigma(\lambda)$ serves as a classification function to separate accepted and rejected proposals. We can generate a new random function $\lambda^*$ by a standard MCMC methods for GPs, for example, elliptical slice sampling~\citep{murray2010elliptical} or Hamiltonian Monte Carlo~\citep{neal2011mcmc}.
    
\item[Sample $\theta,\phi|\,\lambda, \mathcal{Y}^1, \dots, \mathcal{Y}^n, (X_{A}^1, X_{A^c}^1) \dots, (X_{A}^n, X_{A^c}^n)$:] Given all other variables, $\theta$ and $\phi$ are dependent on each other, and follow distributions 
\begin{align}
    p(\phi\,|\,\theta,\mathcal{Y}^1, \dots, \mathcal{Y}^n,  (X_{A}^1, X_{A^c}^1), \dots, (X_{A}^n, X_{A^c}^n)) &\propto p_\phi(\phi)\cdot \prod_{i=1}^n p(X_A^i|\phi) \cdot \prod_{i=1}^n \left\{\pi_0(X_{A^c}^i|X_A^i,\theta,\phi)\prod_{j=1}^{|\mathcal{Y}^i|}\pi_0(y^{ij}|X_A^i,\theta,\phi)\right\}, \label{eq:phi} \\
p(\theta\,|\,\phi, \mathcal{Y}^1, \dots, \mathcal{Y}^n,  (X_{A}^1, X_{A^c}^1), \dots, (X_{A}^n, X_{A^c}^n)) &\propto p_\theta(\theta) \cdot \prod_{i=1}^n \left\{\pi_0(X_{A^c}^i|X_A^i,\theta, \phi)\prod_{j=1}^{|\mathcal{Y}^i|}\pi_0(y^{ij}|X_A^i,\theta,\phi)\right\}.\label{eq:theta}
\end{align}
These can be simulated using standard MCMC techniques such as Metropolis-Hastings, Hamiltonian Monte Carlo or slice sampling.
A simplifying assumption is to make the centering distribution $\pi_0$ independent of the parameter $\phi$. 
While we do not make this assumption, it does hold in the setting of distribution constraints where $\phi$ is fixed.
In such instances, it might be possible to choose a prior $p_\phi(\cdot)$ that is conjugate to the constraint family and a prior $p_\theta(\cdot)$ conjugate to the centering distribution.
Then, given the rejected samples, the posterior belongs to the same family as the prior and is typically easy to sample from.
\end{description}

\begin{algorithm}[H]
  \KwIn{The observations $\mathcal{X}=\left\{(X_{A}^1, X_{A^c}^1), \dots, (X_{A}^n, X_{A^c}^n)\right\}$, set of rejected samples $\Tilde{\mathcal{Y}}=\{(X_{A}^1, \Tilde{y}^{11}),\dots,(X_{A}^1, \Tilde{y}^{1|\Tilde{\mathcal{Y}}^1|}),\dots, (X_{A}^n,\Tilde{y}^{n1}),\dots,(X_{A}^n,\Tilde{y}^{n|\Tilde{\mathcal{Y}}^n|}) \}$, where $\Tilde{\mathcal{Y}}^1=\{\Tilde{y}^{11},\dots, \Tilde{y}^{1|\Tilde{\mathcal{Y}}^1|}\},\dots, \Tilde{\mathcal{Y}}^n=\{\Tilde{y}^{n1},\dots, \Tilde{y}^{n|\Tilde{\mathcal{Y}}^n|}\}$, the current random function values  $\lambda_{\mathcal{X}\cup \Tilde{\mathcal{Y}}}$, the current parameters $\Tilde{\phi}$ and $\Tilde{\theta}$. }
  \KwOut{ A new set of rejected samples $\mathcal{Y}=\{(X_{A}^1, y^{11}),\dots,(X_{A}^1, y^{1|\mathcal{Y}^1|}), \dots,(X_A^n,y^{n1}), \dots,(X_{A}^n,y^{n|\mathcal{Y}^n|}) \}$, where $\mathcal{Y}^1=\{y^{11},\dots, y^{1|\mathcal{Y}^1|}\}, \dots, \mathcal{Y}^n=\{y^{n1},\dots, y^{n|\mathcal{Y}^n|}\}$, and a new instantiation $\lambda_{\mathcal{X}\cup\mathcal{Y}}$ of the GP on $\mathcal{X}\cup\mathcal{Y}$, new parameters $\phi$ and $\theta$.}
 
 Set $S\gets\{1,\dots, n\}$, $\lambda_{pre}=\lambda_{\mathcal{X}\cup\mathcal{\Tilde{Y}}}$, and $\mathcal{Y}\gets\{\}$.
 
 \SetAlgoLined\SetArgSty{}
 \Repeat{$S$ is null}
 {
  Sample $y^{i}\sim \pi_0(\cdot|X_{A}^i; \Tilde{\theta}, \Tilde{\phi})$ $\forall i\in S$, and denote $T=\{(X_{A}^i,y^i), i\in S \}$;
 
 Sample $\lambda_T|\lambda_{pre}$ conditionally from the Gaussian process, and add $\lambda_T$ into $\lambda_{pre}$;
 
 Set $\alpha^i=1$ with probability $\sigma(\lambda(X_{A}^i, y^{i}))$ $\forall i \in S$;
 
 Set $\mathcal{Y}\gets\mathcal{Y}\cup(X_{A}^i,y^i)$ if $\alpha^i=0$; otherwise, delete index $i$ from $S$ $\forall i \in S$.
 }
 
 Restrict $\lambda_{pre}$ to $\lambda_{\mathcal{X}\cup\mathcal{Y}}$.
  
  Update $\lambda_{\mathcal{X}\cup\mathcal{Y}}$ with a Markov kernel having a stationary distribution as~\cref{eq:lambda}.

Update $\Tilde{\phi}$ to $\phi$ with a Markov kernel having a stationary distribution as~\cref{eq:phi}.

Update $\Tilde{\theta}$ to $\theta$ with a Markov kernel having a stationary distribution as~\cref{eq:theta}.
   \caption{An iteration of the Markov chain for posterior inference for $p(\lambda,\phi, \theta|X_{A}^1, X_{A^c}^1, \dots, X_{A}^n, X_{A^c}^n)$}
  
  \label{algo:3}
\end{algorithm}

\section{Experiments\label{sec:3}}

In this section, we present two synthetic examples and two real examples to demonstrate the usefulness of our proposed methodology. For the first synthetic example, we incorporate a specific marginal {\em distribution} constraint into the model and assume the empirical marginal distribution of observations agree with the marginal distribution constraint. The second example demonstrates the additional flexibility of our approach over standard conditional density modeling, by including an additional coupling between the variances of the constrained and unconstrained coordinates.
For both real examples, we incorporate marginal {\em family} constraints into the models.
 We implemented our algorithm in Python and ran it on Purdue Community Cluster Workbench, an interactive compute environment for non-batch big data analysis and simulation, consisting of Dell compute nodes with 24-core AMD EPYC 7401P processors (24 cores per node), and 512 GB of memory. 
 For all examples, we ran a total of 5000 MCMC iterations, and treated the first 1000 iterations as burn-in. 
 
 We compared our proposed marginally constrained model with a fully nonparametric model without any marginal constraints, specifically, the model of~\citet{adams} that ours is based upon.
 For both models, we used Gaussian processes with a squared exponential kernel: $k_{SE}(x,x')=\sigma^2\exp\left({-\cfrac{\|x-x'\|^2}{2l^2}}\right)$.
 In our experiments, we set the parameter $\sigma^2$ to $1$ and updated lengthscale parameter $l$ via Hamiltonian Monte Carlo (HMC) under a weakly informative prior. 
{In the family constraint setting, we also compared our proposed marginally constrained model with a simple parametric model that satisfies the marginal family constraint.} 
 We carried out both quantitative and qualitative evaluations of the models; for the former, we used the likelihood of a held-out test dataset. 

\subsection{Synthetic Example 1}

\begin{figure}
  \centerline{
  \includegraphics[width=0.33\textwidth]{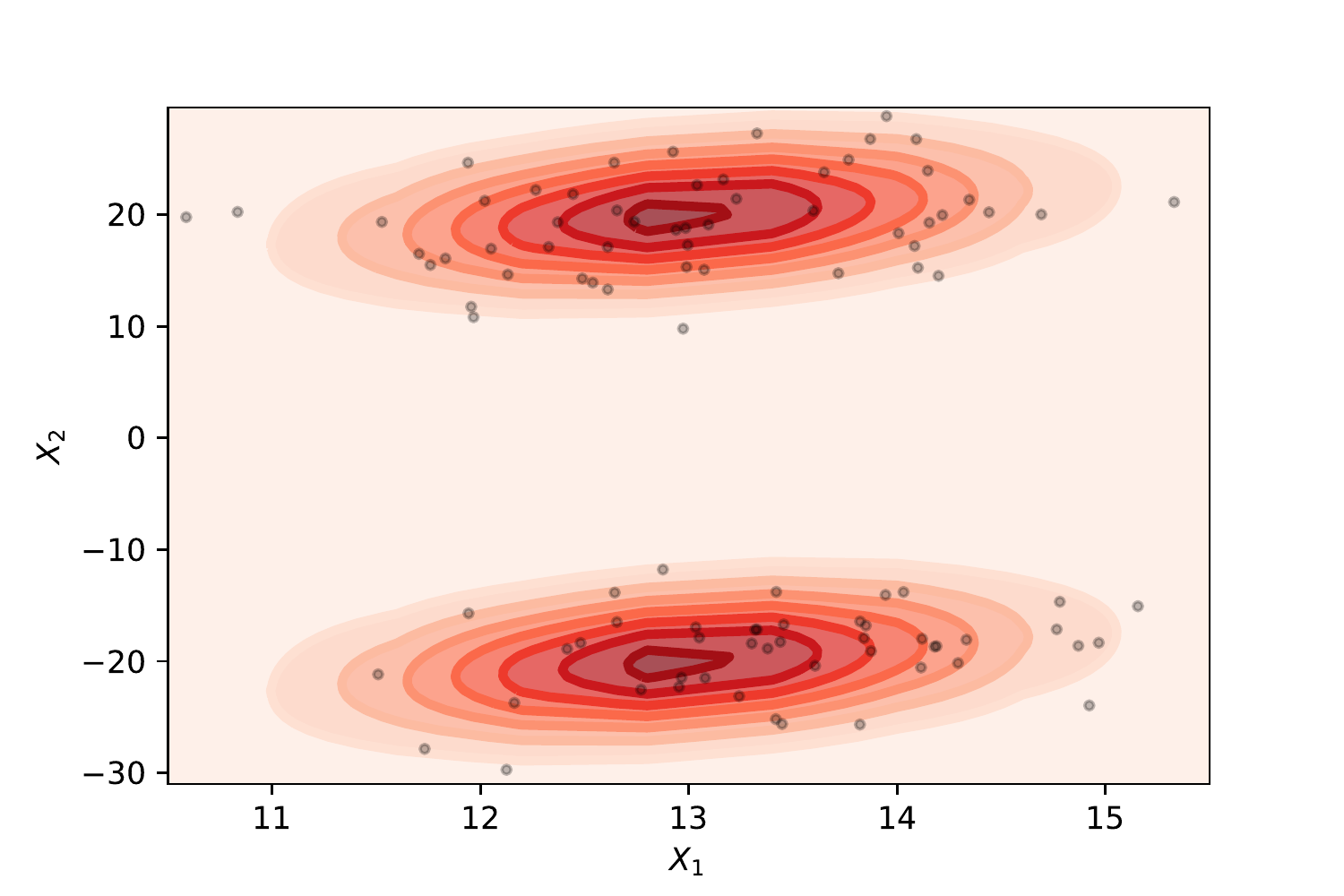}
  \includegraphics[width=0.33\textwidth]{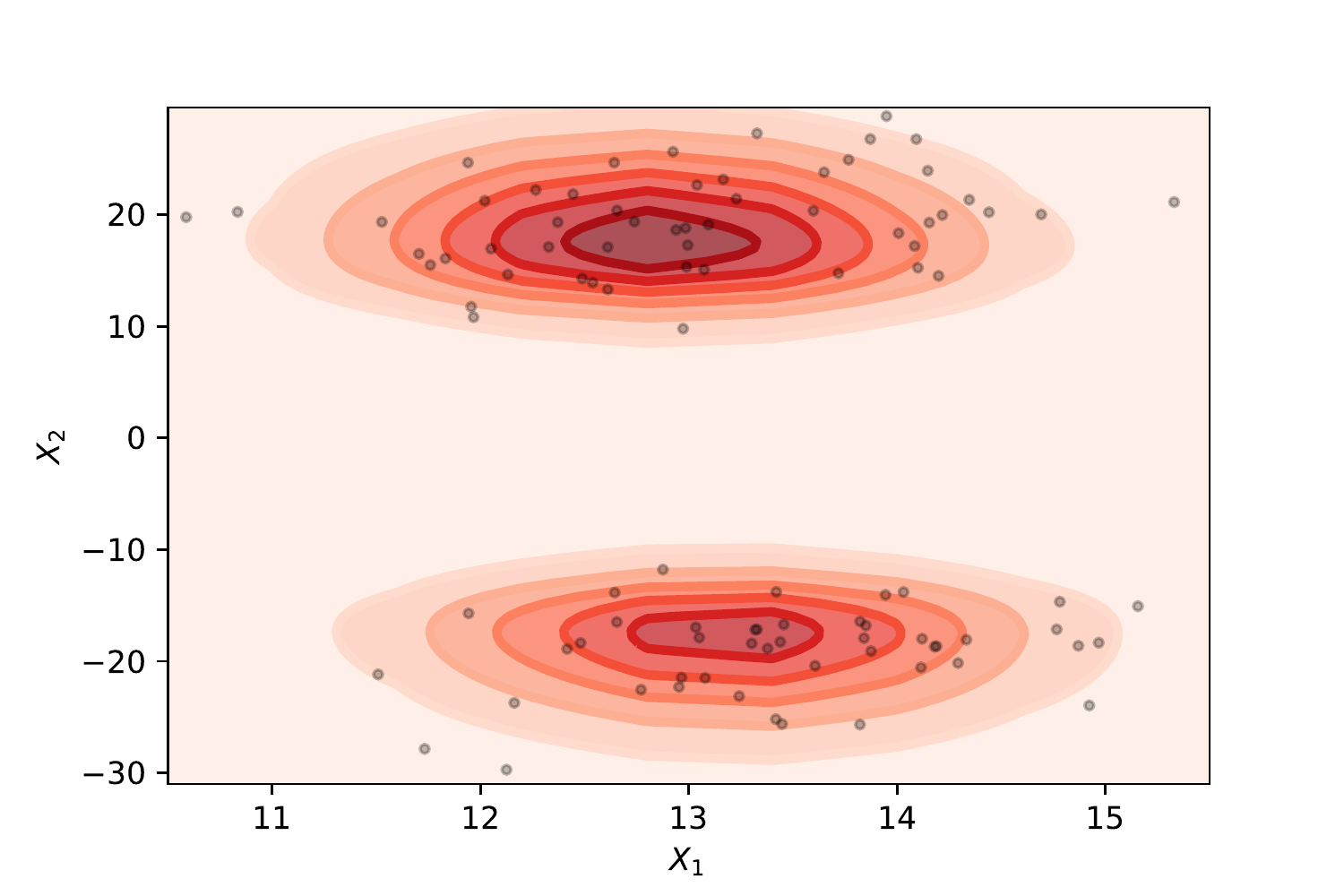}
  \includegraphics[width=0.33\textwidth]{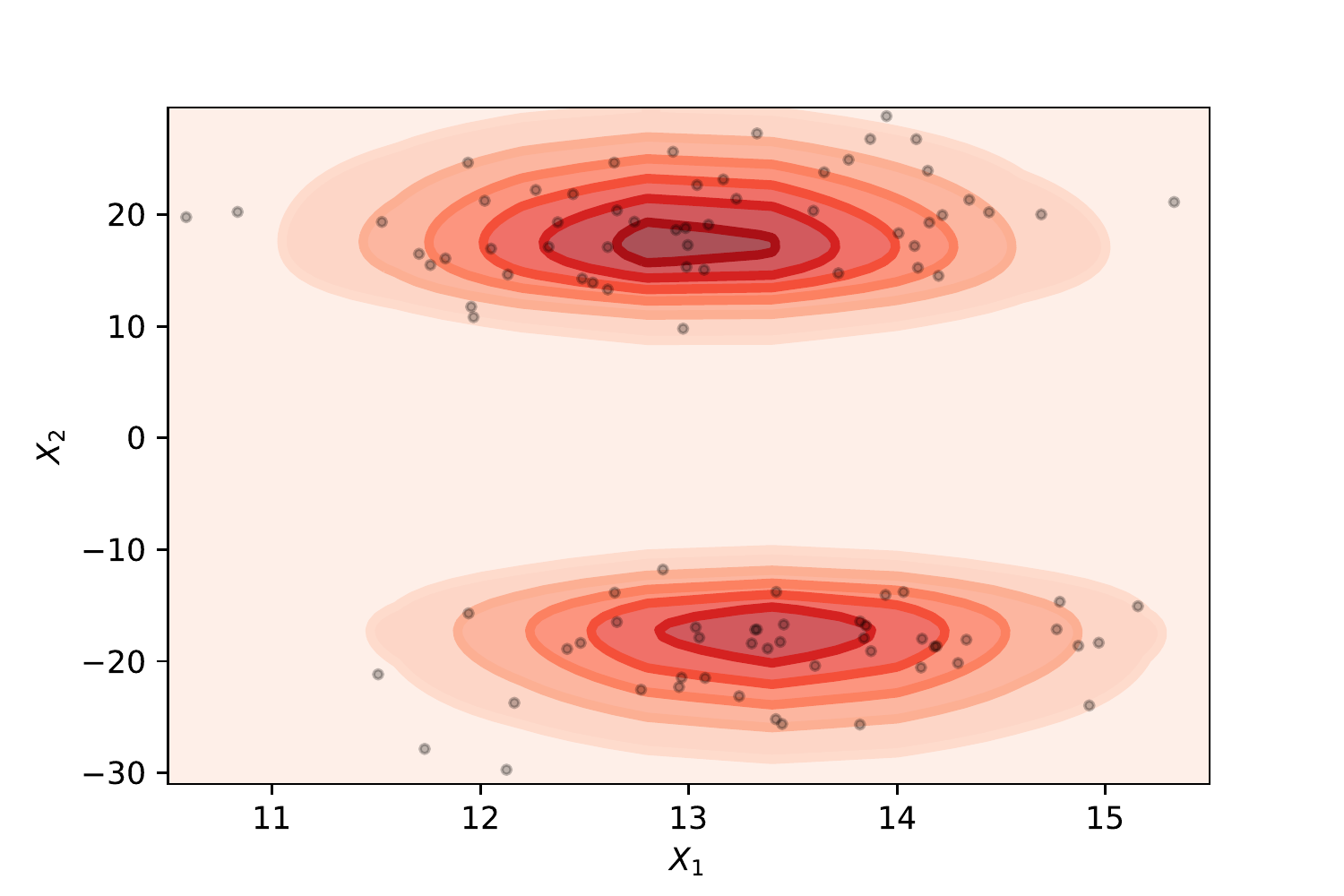}
  }
  
    \centerline{
  \includegraphics[width=0.33\textwidth]{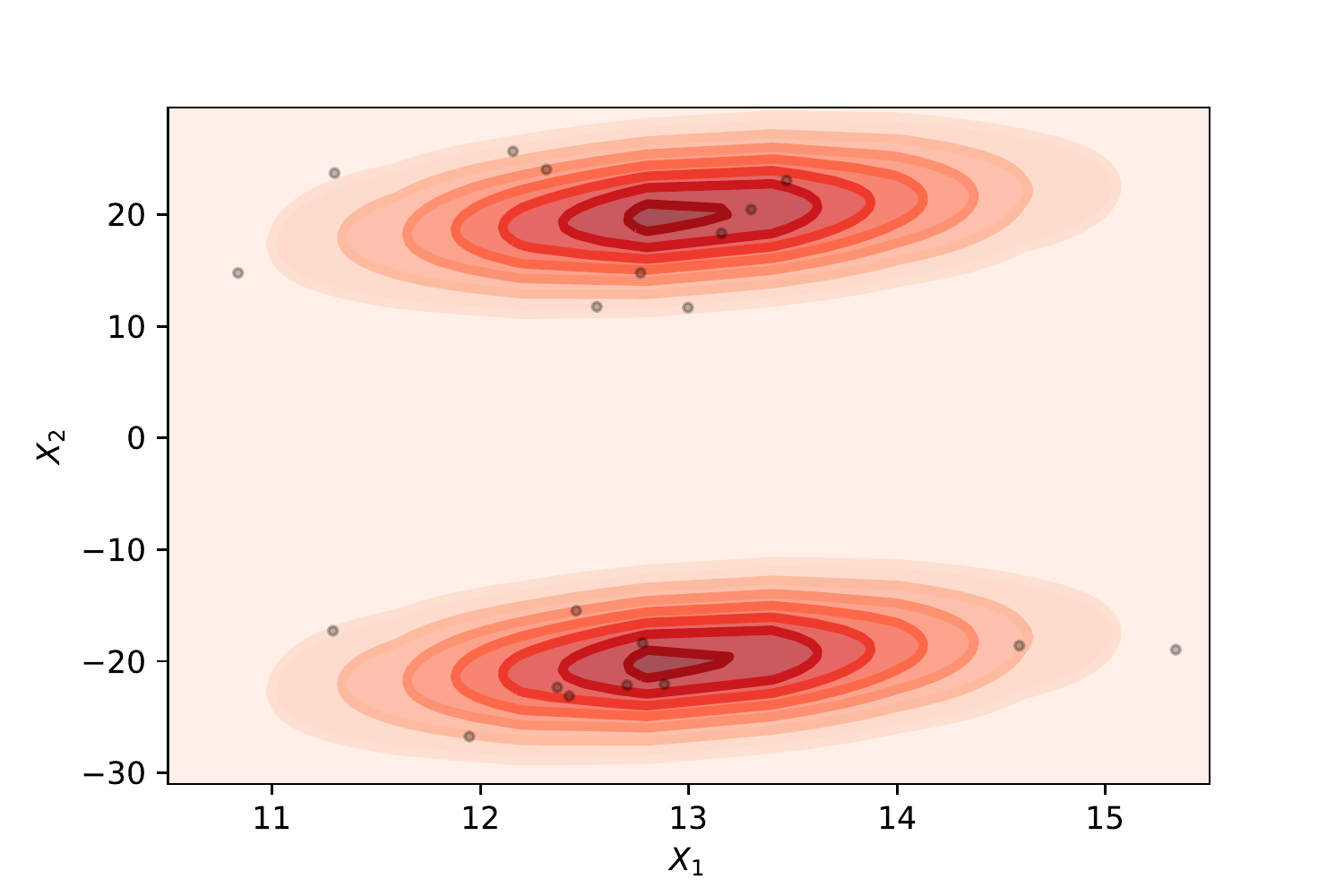}
  \includegraphics[width=0.33\textwidth]{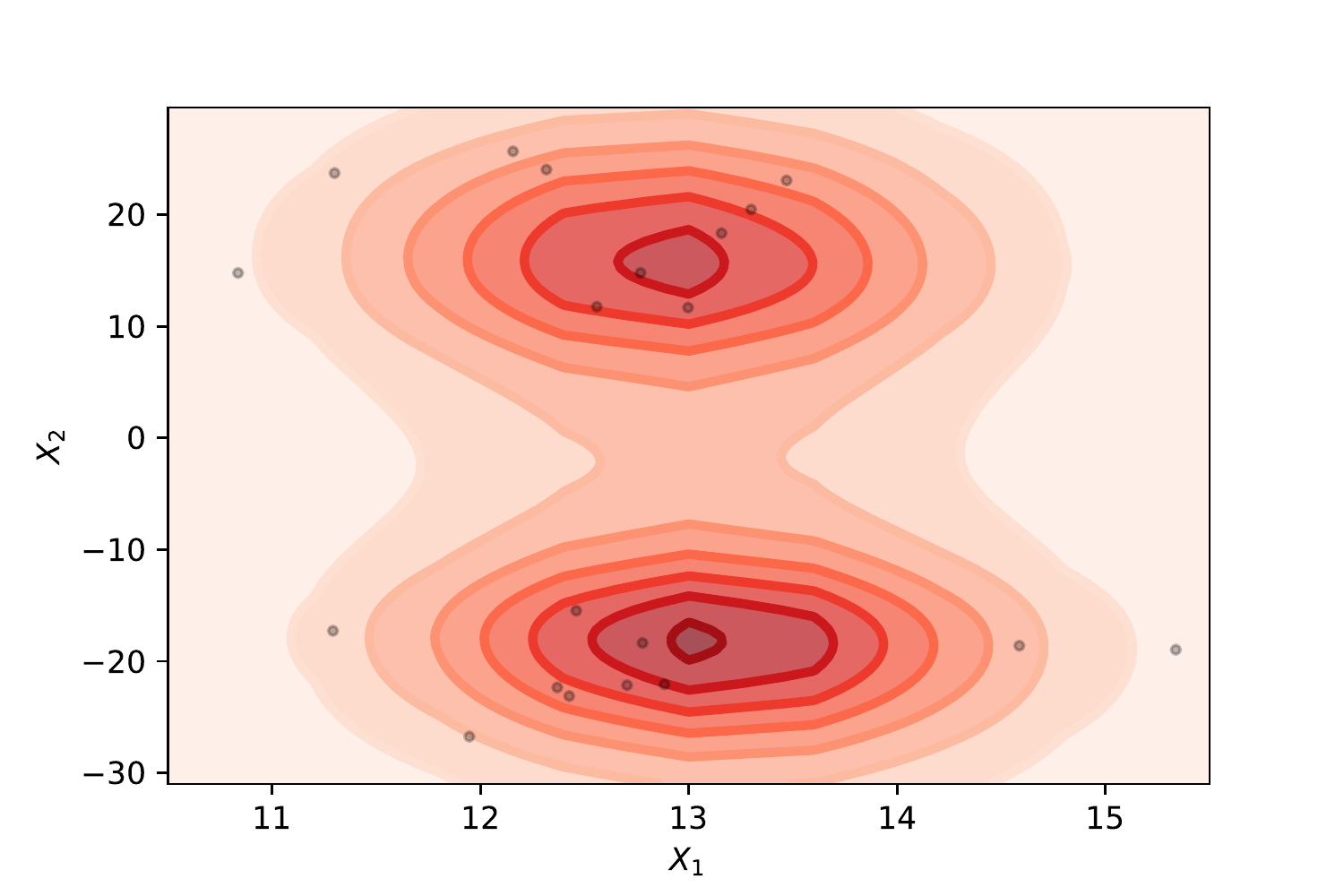}
  \includegraphics[width=0.33\textwidth]{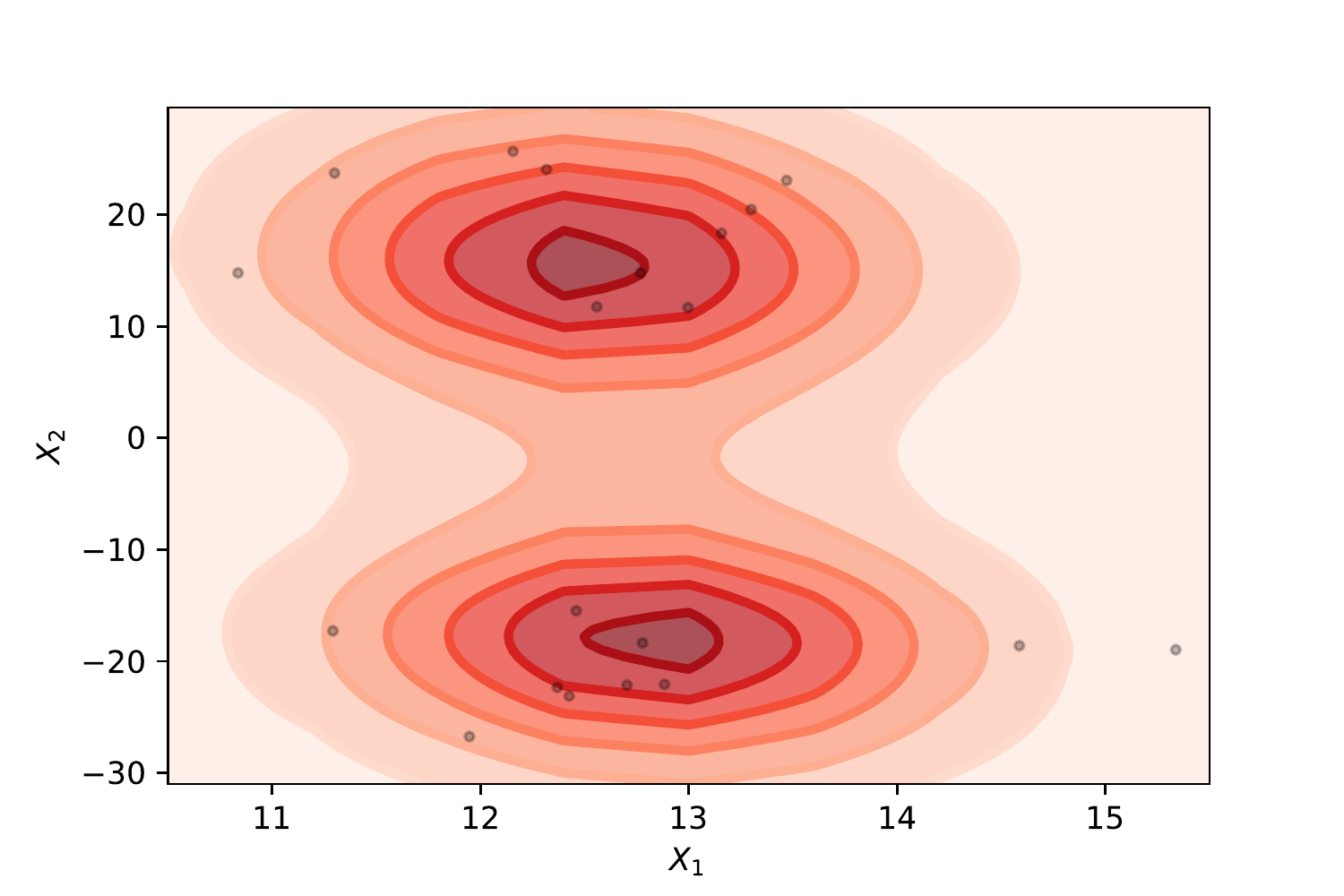}
  }
  
  \caption{(top left): the groundtruth density $0.5\,\mathcal{N}\left(\begin{pmatrix}
13 \\
-20 \\
\end{pmatrix},\begin{pmatrix}
1 & 
\frac{3\sqrt{5}}{5}\\
\frac{3\sqrt{5}}{5} &20 \\
\end{pmatrix}\right)+0.5\,\mathcal{N}\left(\begin{pmatrix}
13 \\
20 \\
\end{pmatrix},\begin{pmatrix}
1 & \frac{3\sqrt{5}}{5}\\
\frac{3\sqrt{5}}{5} &20 \\
\end{pmatrix}\right)$; (top middle): the posterior mean density based on 100 observations drawn from the true density using our proposed model; (top right): the posterior mean density based on 100 observations drawn from the true density with posterior samples of lengthscale parameter via HMC using the fully nonparametric model. The bottom panels are similar, now based on 20 observations drawn from the true density.}\label{fig:normalinvgamma}
\end{figure}
\begin{figure}
  \centerline{
  \includegraphics[width=0.4\textwidth]{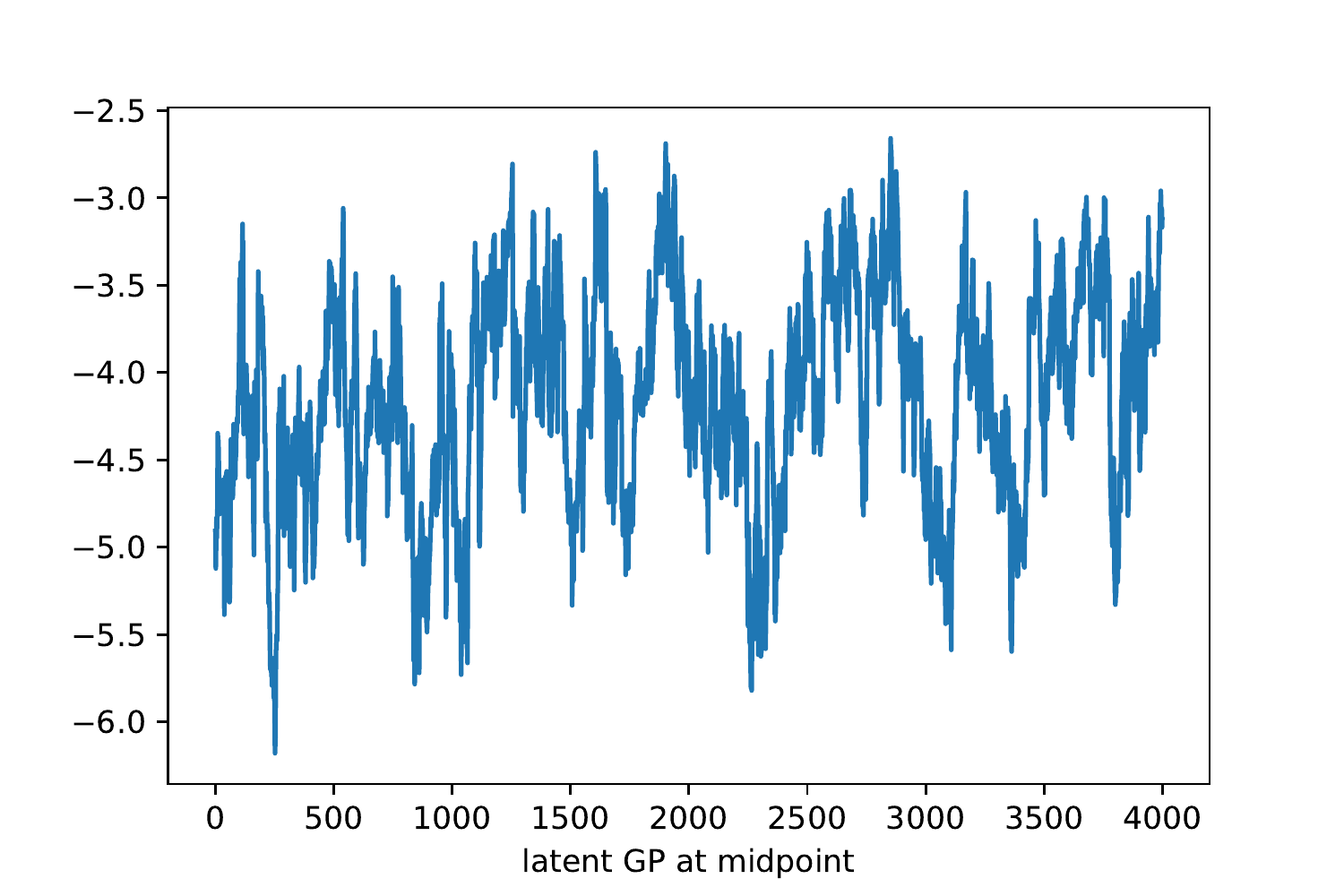}
  \includegraphics[width=0.4\textwidth]{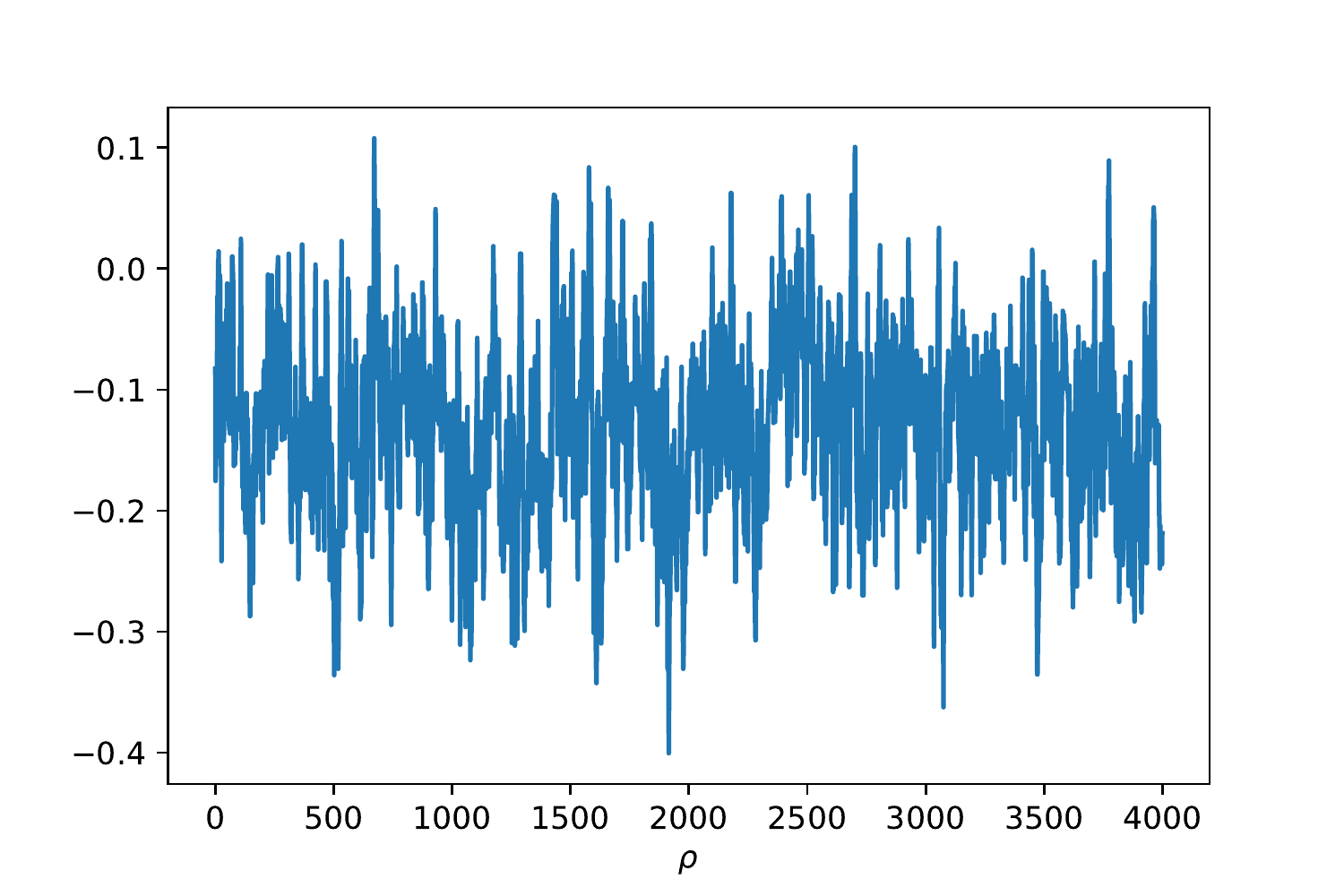}
  \includegraphics[width=0.4\textwidth]{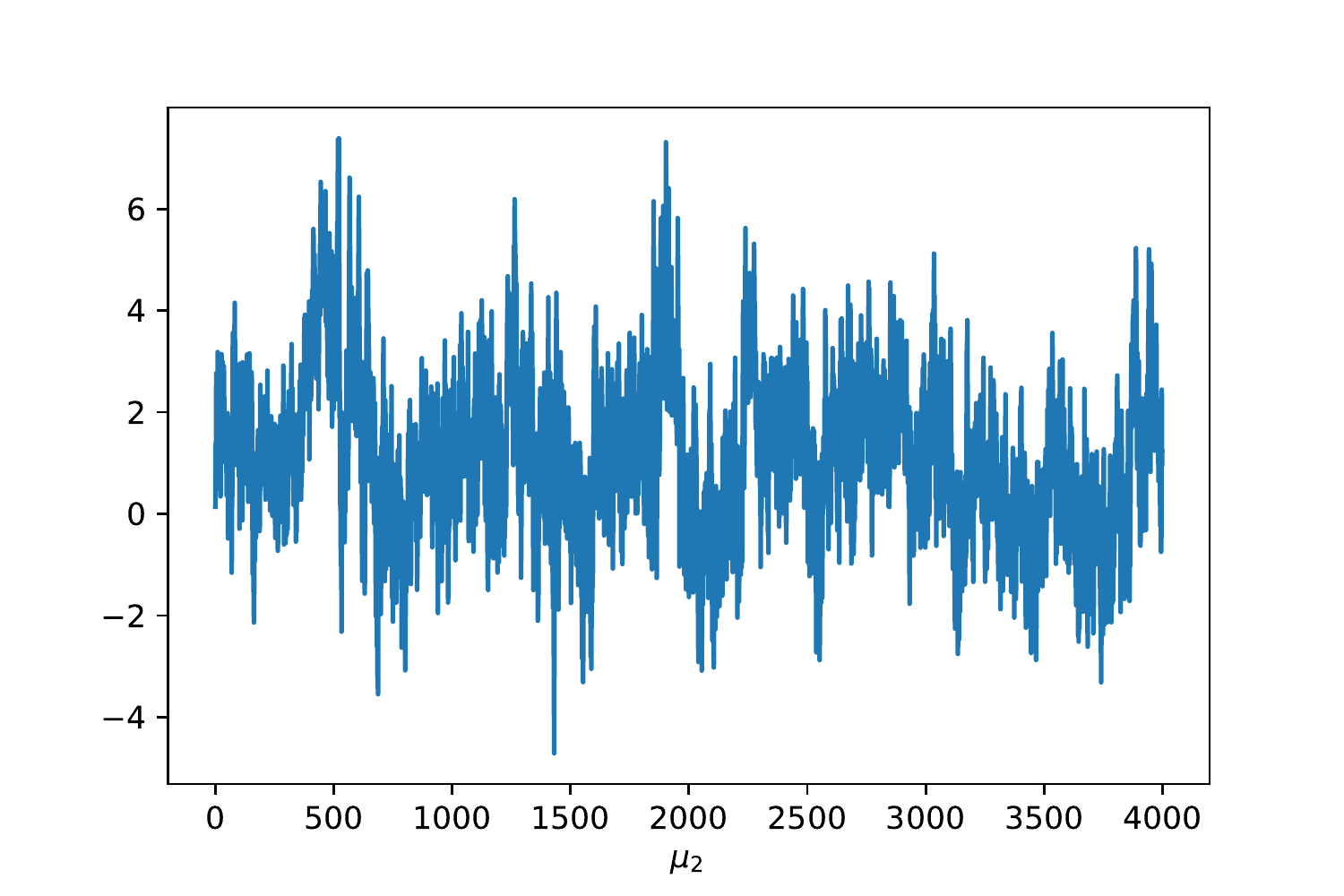}
  }
  \caption{MCMC traceplots for posterior samples of the latent GP at the midpoint along with parameters $\rho$ and $\mu_2$ for the first synthetic example.}\label{fig:traceplotssyn2}
\end{figure}
 Here, we generated 2-dimensional datasets of size 100 and 20 from the following mixture of two Gaussians: $$ (X_1, X_2)\sim 0.5\,\mathcal{N}\left(\begin{pmatrix}
13 \\
-20 \\
\end{pmatrix},\begin{pmatrix}
1 & 
\frac{3\sqrt{5}}{5}\\
\frac{3\sqrt{5}}{5} &20 \\
\end{pmatrix}\right)+0.5\,\mathcal{N}\left(\begin{pmatrix}
13 \\
20 \\
\end{pmatrix},\begin{pmatrix}
1 & \frac{3\sqrt{5}}{5}\\
\frac{3\sqrt{5}}{5} &20 \\
\end{pmatrix}\right).$$ 
The leftmost panel of~\cref{fig:normalinvgamma} shows the contours of this probability density with different sets of observations.
Observe that the first coordinate of these observations follows a Gaussian distribution 
$p_1(X_1)=\mathcal{N}(X_1\,|\,\mu_1,\sigma_1^2),$
where $\mu_1=13$ and $\sigma_1^2=1$.
We assume this marginal constraint is known, and model 
{each} dataset with our marginally constrained nonparametric prior.
We choose the centering distribution $\pi_0(X_2|X_1;\theta, \phi)$ as a conditional normal distribution, namely $$\pi_0(X_2|X_1=x_1;\theta, \phi)=\mathcal{N}\left(X_2 \,;\mu_2+\rho\sigma_2 \sigma_1^{-1}(x_1-\mu_1), \left(1-\rho^2\right)\sigma_2^2\right)$$ where $\theta=(\rho, \mu_2, \sigma_2^2)$, and $\phi=(\mu_1, \sigma^2_1)$, which are known values here. 
We place a Normal-Inverse-Gamma-Uniform prior on $\theta=(\rho, \mu_2, \sigma_2^2)$, with $\mu_0=0$, $k=0.001$, $\alpha_0=0.001$ and $\beta_0=0.001$: 

$$p(\rho, \mu_2, \sigma_2^2)\propto \mathcal{N}\left(\mu_2\,;\mu_0,\cfrac{\sigma_2^2}{k_0}\right)\cdot\text{ Inv-Gamma} (\sigma_2^2\,;\alpha_0, \beta_0)\cdot\mathbbm{1}_{[-1,1]}(\rho).$$ 
Running our MCMC sampler from~\cref{algo:3} (including a Hamiltonian Monte Carlo update for the lengthscale parameter in the kernel covariance matrix), we produce 4000 posterior samples for $\theta$ and $\lambda$ for each of the two datasets   and then use those posterior samples to compute the mean of data densities, which is presented in the middle column of~\cref{fig:normalinvgamma}. {We include the traceplots for the posterior samples corresponding to the dataset of size 100 in figure~\ref{fig:traceplotssyn2}. For comparison, we also compute the mean of data densities for the fully nonparametric model, which is displayed at the rightmost column of~\cref{fig:normalinvgamma}. We can observe that the fully nonparametric model does not satisfy the marginal distribution constraint exactly, which is further supported in the quantitative comparison below.}

To quantitatively compare posterior results of our proposed marginally constrained model and the fully nonparametric model, we generate a test dataset of size 60 and 5 training datasets of size 20, 60 and 100 respectively from the mixture of normal distribution. 
For each model and each training dataset, 
we produce 4000 posterior samples and then use those posterior samples to compute `marginal' and `joint' loglikelihoods of the test dataset. {Here, the joint loglikehood refers to the standard logarithmic probability of the test dataset, while the marginal loglikelihood describes the logarithmic probability of the first component of the test datapoints (viz.\ the constrained component).   }
{Finally, for each model 
and each training sample size, the median of
average loglikelihoods over posterior samples across the 5 training-test splits is reported in ~\cref{table:marginalsyn}  and~\cref{table:jointsyn}.}
\begin{table*}
\footnotesize
\begin{center}
\caption{Average marginal loglikelihood in the first synthetic example}
\label{table:marginalsyn}
\begin{threeparttable}
 \begin{tabular}
{
  @{\kern-.5\arrayrulewidth}
  |p{\dimexpr3.5cm-4\tabcolsep-.5\arrayrulewidth}
  |p{\dimexpr4.5cm-5\tabcolsep-.5\arrayrulewidth}
  |p{\dimexpr4.5cm-5\tabcolsep-.5\arrayrulewidth}
  |@{\kern-.5\arrayrulewidth}
}
 \hline
  Training dataset size & Truth/Our proposed model  & Fully nonparametric model 
  \\ [0.5ex]
 \hline
 \textbf{100} &\textbf{-82.79}& \textbf{-83.57}\\ 
 \hline
 \textbf{60} &\textbf{-82.79}& \textbf{-84.02}\\
 \hline
 \textbf{20} &\textbf{-82.79}& \textbf{-86.88 }\\
 \hline
\end{tabular}
\end{threeparttable}
\end{center}
\end{table*}
\begin{table*}
\footnotesize
\begin{center}
\caption{ Average joint loglikelihood in the first synthetic example}
\label{table:jointsyn}
\begin{threeparttable}
\begin{tabular}
{
  @{\kern-.5\arrayrulewidth}
  |p{\dimexpr3.5cm-4\tabcolsep-.5\arrayrulewidth}
  |p{\dimexpr1.5cm-2\tabcolsep-.5\arrayrulewidth}
  |p{\dimexpr3.5cm-4\tabcolsep-.5\arrayrulewidth}
  |p{\dimexpr4.5cm-5\tabcolsep-.5\arrayrulewidth}
  |@{\kern-.5\arrayrulewidth}
}
 \hline
  Training dataset size & Truth  & Our proposed model & Fully nonparametric model\\ [0.5ex]
 \hline
 
 \textbf{100} &\textbf{-299.63} &\textbf{-304.69}  &\textbf{-305.67}\\
 \hline
 \textbf{60} &\textbf{-299.63}& \textbf{-306.70} &\textbf{-308.64}\\
 \hline
 \textbf{20} &\textbf{-299.63} &\textbf{-318.41}  &\textbf{-322.40} \\
 \hline
\end{tabular}
\end{threeparttable}
\end{center}
\end{table*}

As reported in the two tables, we conclude that, compared with the fully nonparametric model, both joint and marginal loglikelihoods for our proposed marginally constrained model are always closer to the truth. The difference of either joint or marginal loglikelihoods between the two models increases as the size of traning datasets diminishes.

\subsection{Synthetic example 2}
 For this example, we consider a setting that requires a bit more structure than our original model.
Specifically, we assume that two random variables $X_1$ and $X_2$, where $X_1$ is known to follow a normal distribution with unknown parameters $\mu_1$ and $\sigma_1^2$, i.e., $p_1(X_1)=\mathcal{N}(X_1\,;\mu_1,\sigma_1^2).$
While the conditional distribution of $X_2$ is unknown, it is known to have a variance of the same order as $X_1$; this is a reasonable assumption in many settings. We now seek to model a dataset of observations of $(X_1, X_2)$, while incorporating both pieces of information into the joint model. 
Our original model already allows the marginal to be incorporated, and a simple modification to incorporate the variance constraint is by setting 
the centering distribution $\pi_0(X_2|X_1; \theta,\phi)$ to a conditional normal distribution as below:  
$$\pi_0(X_2|X_1=x_1;\theta,\phi)=\mathcal{N}(X_2 \,;\mu_2+\rho(x_1-\mu_1), (1-\rho^2)\sigma_1^2).$$ 
Observe that we use the same variance $\sigma^2_1$ as the marginal constraint. 
We place the normal-inverse-gamma prior on $\theta=(\rho, \mu_2, \sigma_1^2)$, specifically, with $\mu_0=0$, $k_0=0.001$, $\alpha_0=0.001$ and $\beta_0=0.001$, we set
$$P(\rho, \mu_2, \sigma_1^2)    \propto N(\mu_2;\mu_0,\frac{\sigma_1^2}{k_0}) \text{Inv-Gamma}(\sigma_1^2;\alpha_0, \beta_0)\mathbbm{1}_{[-1,1] }(\rho).$$ 
Simultaneously, with  $\mu_{x_0}=-10$ and $k_{x_0}=0.01$, we place the following normal prior on $\phi=\mu_1$:
\begin{align*}
    \mu_1|\sigma_1^2&\sim N(\mu_{x_0},\frac{\sigma_1^2}{k_{x_0}}).
\end{align*}
Next, we draw 15 observations from the normal distribution shown at topleft of~\cref{fig:variance equivalence2}, namely,$$(X_1, X_2)\sim
\mathcal{N}\left(\begin{pmatrix}
13 \\
-5 \\
\end{pmatrix},\begin{pmatrix}
20 & 6\\
6 &20 \\
\end{pmatrix}\right) $$
Observe that this distributions has same marginal variances for each component. We apply the modified model described above, the fully nonparametric model and our original model to the observations and compute the mean of densities estimated from 4000 posterior samples for $(\rho, \mu_2, \sigma_1^2, \mu_1, \lambda)$ as shown in~\cref{fig:variance equivalence2}. 
{Our MCMC sampler here involves a straightforward modification of~\cref{eq:theta} in ~\cref{algo:3}, now having an additional term since the parameter $\sigma_1^2$ now depends on both the marginal distribution of $X_1$ and the prior distribution of $\mu_1$. 
}
The posterior mean densities for all models show that the performance of the modified model surpasses the other two models with respect to the similarity with the true density. In absolute terms, the recovered density does a good job approximating the truth, despite there being 15 observations.

Apart from the straight qualitative results, we use the same metric in the first synthetic example, namely the joint loglikelihood, to compare the three different models. We generate 10 pairs of training and test datasets of size 15 from the true normal distribution. The lengthscale parameter is updated  via HMC. For each model and each pair of datasets, we produce 4000 posterior samples and compute the joint loglikelihoods of the test dataset.~\Cref{table:jointsyn3} reports quantiles of joint loglikelihoods among the 10 pairs of datasets and shows that the modified model outpeforms the other two models.

\begin{figure}
  \centerline{
  \includegraphics[width=0.65\textwidth]{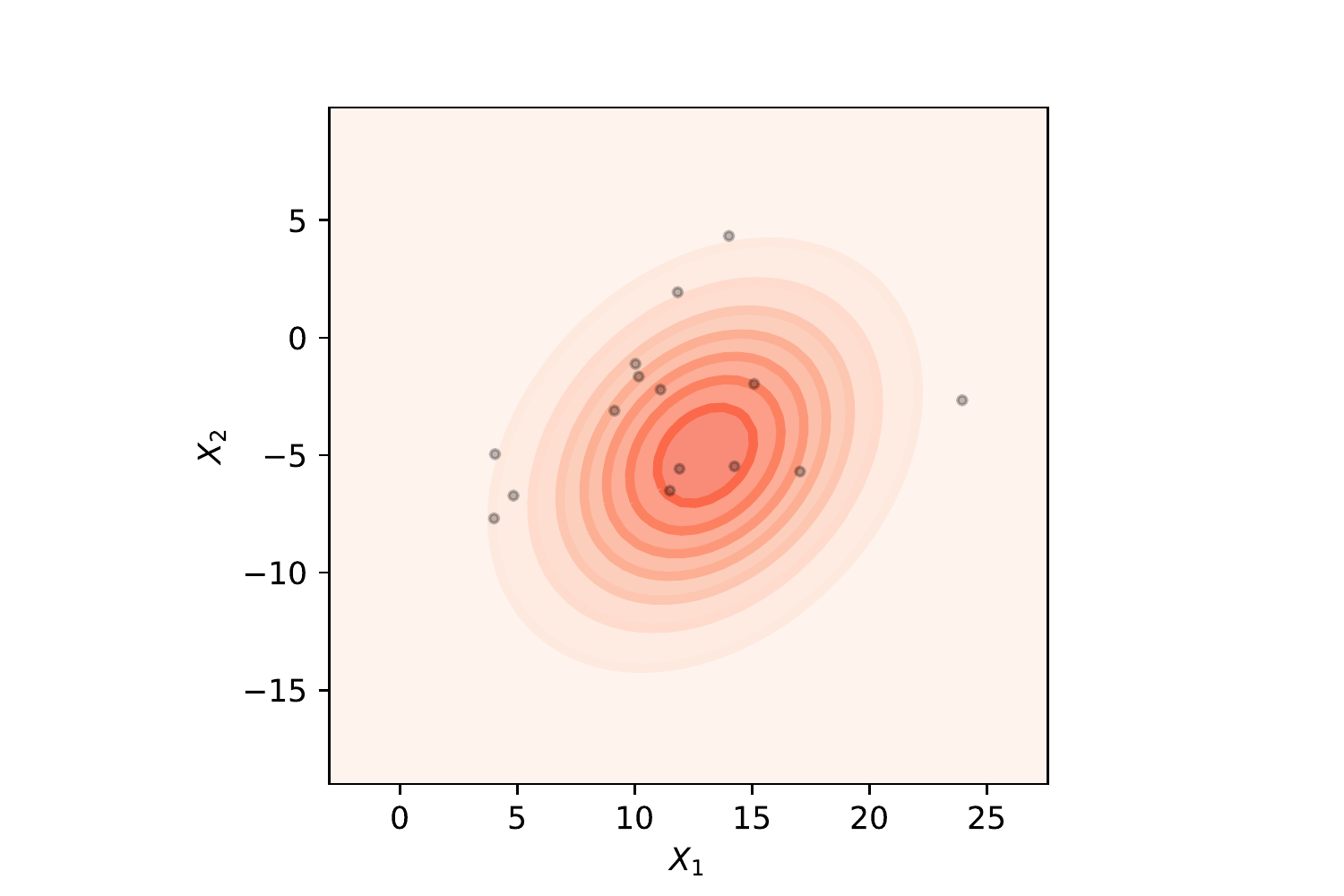}
    \includegraphics[width=0.65\textwidth]{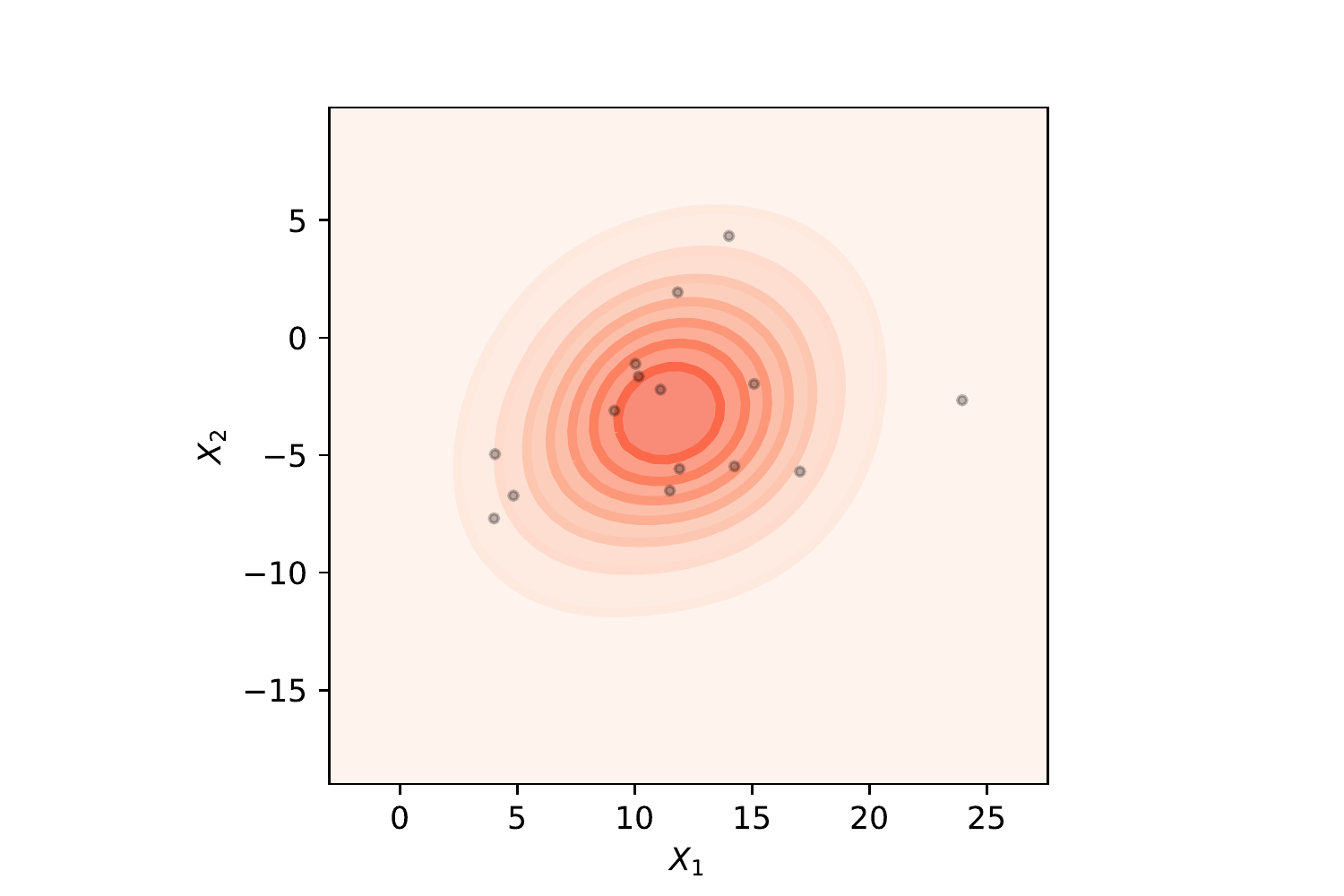}
  
  }
  \centerline{
   \includegraphics[width=0.65\textwidth]{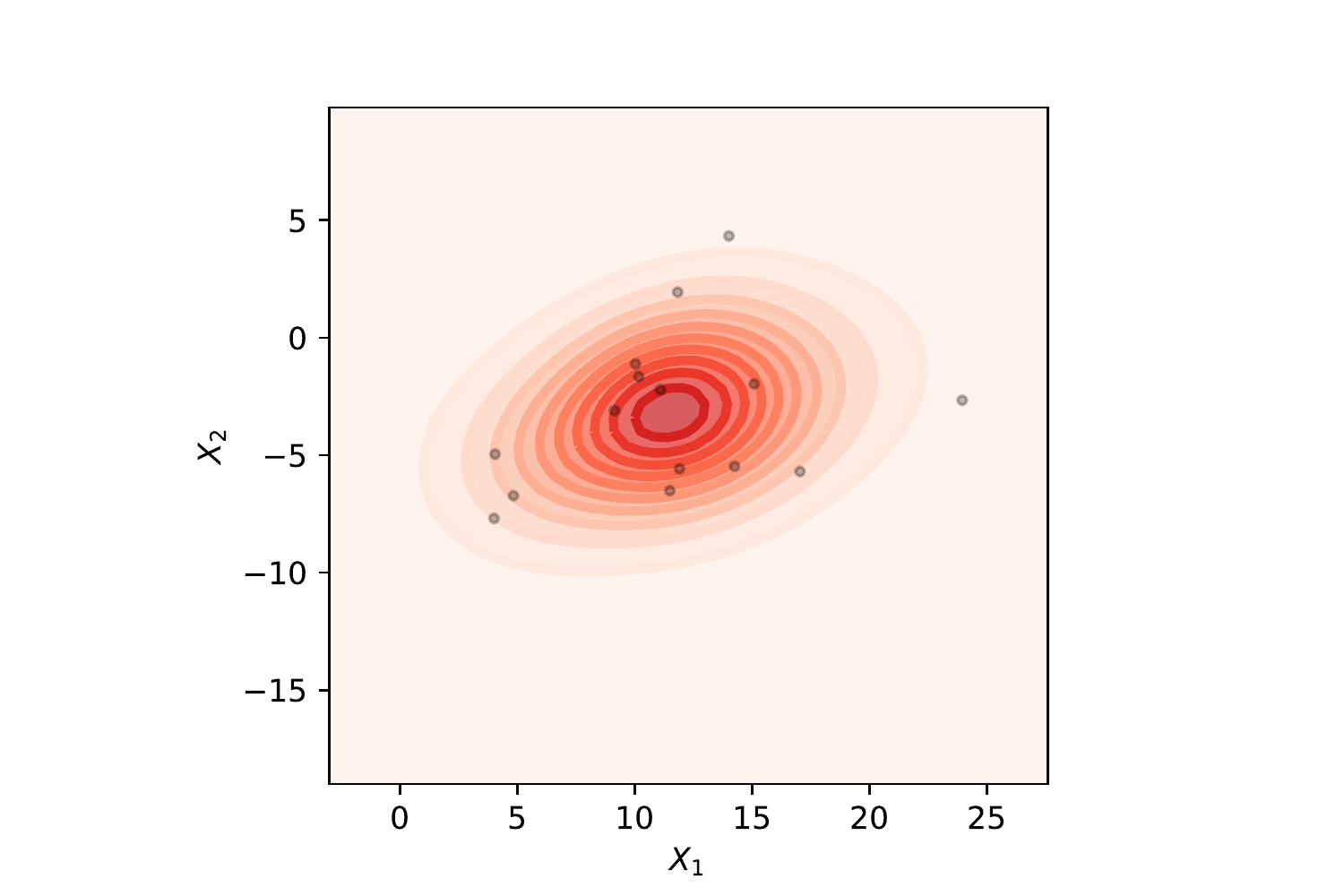}
  \includegraphics[width=0.65\textwidth]{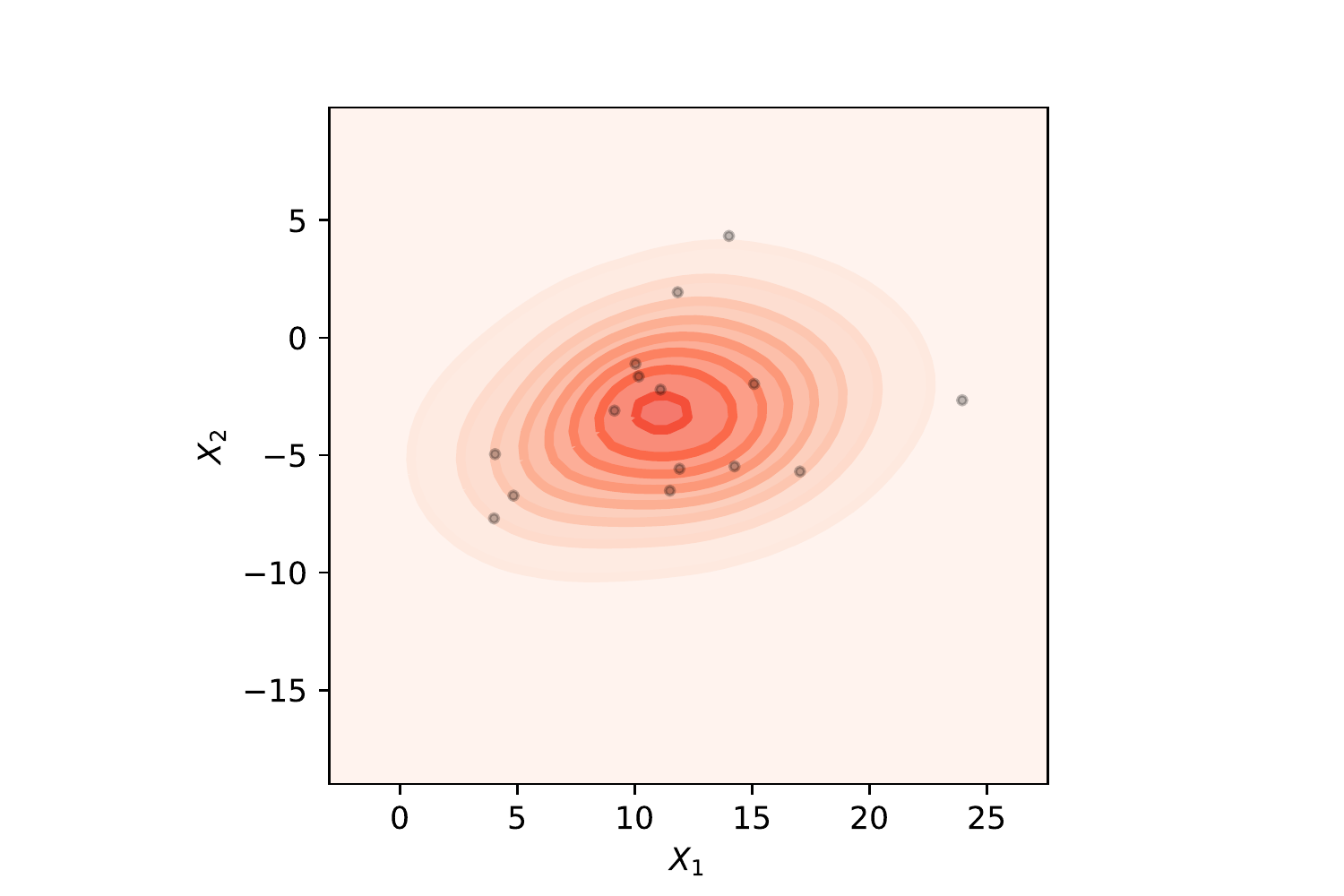}}
  \caption{(top left): the true density $\mathcal{N}\left(\begin{pmatrix}
13 \\
-5 \\
\end{pmatrix},\begin{pmatrix}
20 & 6\\
6 &20 \\
\end{pmatrix}\right) $; (top right): the posterior mean density given 15 observations drawn from the true density using our modified model; (bottom left): the posterior mean density given 15 observations using the fully nonparametric model; (bottom right): the posterior mean density using 15 observations using our original proposed model.} \label{fig:variance equivalence2}
\end{figure} 

\begin{table*}
\footnotesize
\begin{center}
\caption{Average joint loglikelihood in the second synthetic example. The three columns present joint loglikelihoods in the format 0.50 quantile (0.25 quantile, 0.75 quantile) using 10 pairs of training/test datasets.}
\label{table:jointsyn3}
\begin{threeparttable}
 \begin{tabular}
{
  @{\kern-.5\arrayrulewidth}
  |p{\dimexpr4cm-4.5\tabcolsep-.5\arrayrulewidth}
  |p{\dimexpr4.5cm-5\tabcolsep-.5\arrayrulewidth}
  |p{\dimexpr7cm-7.5\tabcolsep-.5\arrayrulewidth}
  |@{\kern-.5\arrayrulewidth}
}
 \hline
 The modified model&  Fully nonparametric model& Our proposed marginally constrained model\\ [0.5ex]
\hline
\textbf{-88.45\,(-90.07,\,-85.68)}&\textbf{-90.25\,(-95.52, \,-87.71)}&\textbf{-89.33\,(-92.03,\, -84.78)}\\
\hline
\end{tabular}
\end{threeparttable}
\end{center}
\end{table*}

\subsection{Real Example 1} \label{sec:pm}
{ Particulate matter 2.5 (PM 2.5) refers to a category of particles in the air that are 2.5 micrometers or less in size~\citep{ott1990physical}. Their ability to penetrate deeply into the lung makes them dangerous to human health.}
It is common to model concentration levels of PM 2.5 with a lognormal  distribution and this concentration level is known to be highly correlated with outdoor temperature; see for example~\citep{ott1990physical}.
We obtain measurements of the daily average concentration level of PM 2.5 and outdoor temperature in Clinton Drive in Houston, TX(CAM 55) for the year 2020\footnote{from the website https://www.tceq.texas.gov/cgi-bin/compliance/monops/yearly\_summary.pl}. 
The original dataset consists of 366 daily observations, which we filtered
to eliminate outliers and missing data points to obtain the final dataset of 356 observations. This is plotted in~\cref{fig:PM2.5dataset}.

 \begin{figure}
  \centerline{
  \includegraphics[width=0.6\textwidth]{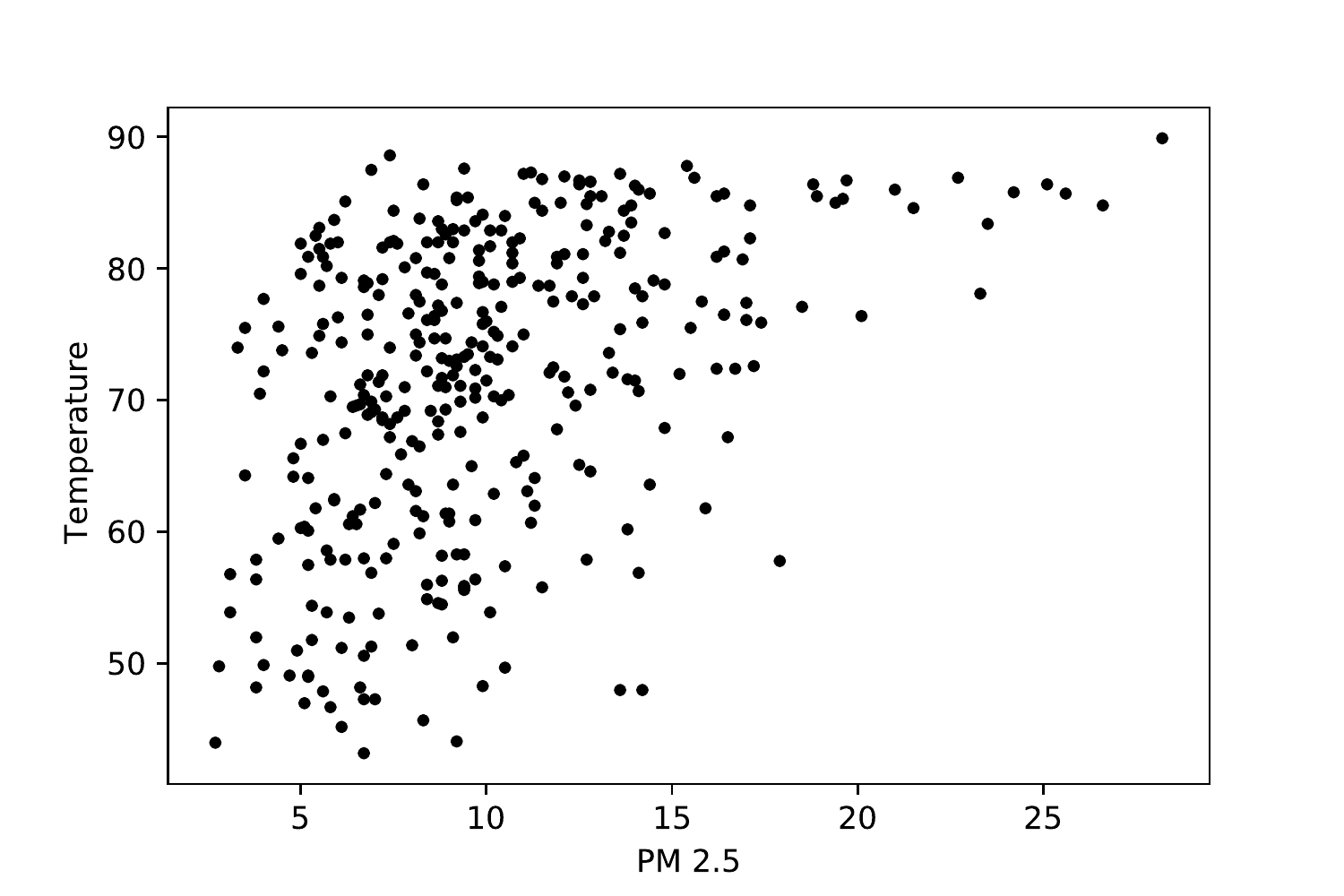}
   \includegraphics[width=0.6\textwidth]{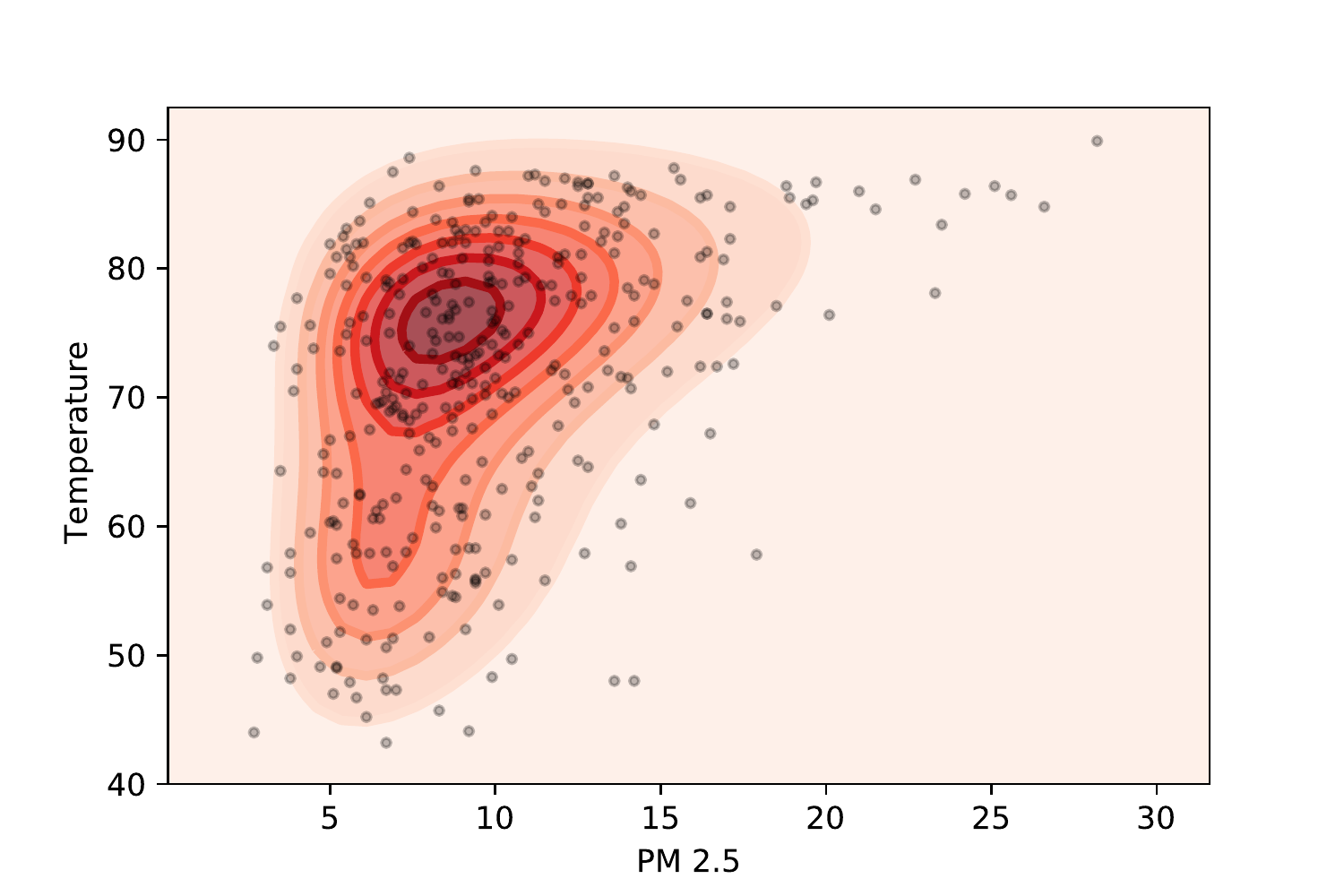}
  }\caption{(left): the PM2.5 dataset showing PM2.5 levels versus temperature; (right): the posterior mean density based on 4000 MCMC samples.
  }\label{fig:PM2.5dataset}
\end{figure}
Let $X_1$ denote the daily average concentration levels of PM 2.5 and $X_2$ denote the daily average outdoor temperature.
We applied our model to this dataset of $(X_1,X_2)$ pairs, imposing a lognormal family constraint on the PM 2.5 concentration levels.
{In this example, a minor modification to the centering distribution is required and an additional prior is placed on the parameters of the family constraint. 
Denoting the parameters of the lognormal family as $\phi=(\mu_x, \sigma_x^2)$, we placed a conjugate normal-inverse-chi-squared prior on these: 
\begin{align*}
    \sigma_x^2&\sim \cfrac{v_{x0}\,\sigma^2_{x0}}{\chi^2_{v_{x0}}}\,, \,  \quad        
    \mu_x\,|\,\sigma_x^2\sim \mathcal{N}(\mu_{x0}\, ,\, \cfrac{\sigma_x^2}{k_{x0}}).
\end{align*} 
We set $\mu_{x0}=-10, k_{x0}=0.01,v_{x0}=0.001$ and $\sigma^2_{x0}=5$.
For our centering distribution, we use 
\begin{align}
\pi_0(X_2|X_1=x_1, \theta, \phi)=\mathcal{N}\left(X_2 \,;\mu_2+\rho\sigma_2 s_x^{-1} \left(x_1- m_x   \right), \left(1-\rho^2\right)\sigma_2^2\right), \label{eq:centering distribution}
\end{align}
where $\theta=(\rho, \mu_2, \sigma_2^2)$, $s_x^2=\Var\left[X_1\right]=\exp\left\{(\sigma_x^2-1)(2\mu_x+\sigma_x^2)\right\}$, and $m_x=\E\left[X_1\right]=\exp\left\{\mu_x+\frac{\sigma_x^2}{2}\right\}$.
We place a normal-inverse-gamma prior on $\theta=(\rho, \mu_2, \sigma_2^2)$: $$P(\rho, \mu_2, \sigma_2^2)\propto \mathcal{N}\left(\mu_2\,;\mu_0,\cfrac{\sigma_2^2}{k_0}\right)\cdot\text{ Inv-Gamma} (\sigma_2^2\,;\alpha_0, \beta_0)\cdot\mathbbm{1}_{[-1,1]}(\rho)$$
where $\mu_0=0$, $k=0.001$, $\alpha_0=0.001$ and $\beta_0=0.001$. 
}
Using our MCMC sampler with this model, we draw samples from the posterior distribution given the PM 2.5 dataset, plotting the posterior mean density in~\cref{fig:PM2.5dataset}. {We see that the posterior mean density captures the characteristics of the dataset reasonably well.}
The model does struggle to capture some of the outliers along the $X_1$-component, though this is more a reflection of the marginal lognormal constraint, rather than the nonparametric component.
Modeling both components together allows practitioners to assess this limitation for different values of the temperature variable, and the figure suggests that failures of the parametric assumption occur at large values of the temperature variable.

Analogous to the quantitative comparison in the first synthetic example, we also perform a comparison among our proposed marginally constrained model, a fully nonparametric model and a parametric model. 
For the fully nonparametric model and our proposed model, the lengthscale parameter is 
updated via HMC. 
We also fit a bivariate normal parametric model to fit temperature and the log-transformed PM 2.5 variable; note that the parametric model satisfies the lognormal family constraint on PM 2.5. We repeat splitting the dataset into a training dataset of size 296 and a test dataset of size 60 5 times and obtain 5 pairs of training and test datasets. For each model 
and each pair of training and test datasets, we produce 4000 posterior samples according to the matching training dataset and then use those posterior samples to compute the joint and marginal loglikelihoods of the corresponding test dataset. Finally, for each model
, the median of the average loglikelihoods over posterior samples across the 5 pairs of training and testing datasets are reported in~\cref{table:jointreal} and~\cref{table:marginalreal}.
Both tables illustrate that our proposed marginally constrained model always behaves the best, {demonstrating the importance of flexibility in preserving the dependence structure between the two variables as well as incorporating prior information through marginal constraints in data-poor settings.}
 
\begin{table*}
\footnotesize
\begin{center}
\caption{Average joint loglikelihood in the PM 2.5 data from 5 training-test splits}
\label{table:jointreal}
\begin{threeparttable}
 \begin{tabular}
{
  @{\kern-.5\arrayrulewidth}
  |p{\dimexpr3.8cm-4.3\tabcolsep-.5\arrayrulewidth}
  |p{\dimexpr4.5cm-5\tabcolsep-.5\arrayrulewidth}
  |p{\dimexpr7cm-7.5\tabcolsep-.5\arrayrulewidth}
  |p{\dimexpr3cm-3.5\tabcolsep-.5\arrayrulewidth}
  |@{\kern-.5\arrayrulewidth}
}
 \hline
  Lengthscale parameter&  Fully nonparametric model& Our proposed marginally constrained model &  Parametric model\\ [0.5ex]
 \hline
 \textbf{HMC}&\textbf{-385.57} &\textbf{-382.24} &\textbf{-384.52} \\
 \hline
\end{tabular}
\end{threeparttable}
\end{center}
\end{table*}

\begin{table*}
\footnotesize
\begin{center}
\caption{ Average marginal loglikelihood for the PM 2.5 data for 5 training-test splits}
\label{table:marginalreal}
\begin{threeparttable}
 \begin{tabular}
{
  @{\kern-.5\arrayrulewidth}
  |p{\dimexpr3.8cm-4.3\tabcolsep-.5\arrayrulewidth}
  |p{\dimexpr4.5cm-5\tabcolsep-.5\arrayrulewidth}
  |p{\dimexpr7cm-7.5\tabcolsep-.5\arrayrulewidth}
  |p{\dimexpr3cm-3.5\tabcolsep-.5\arrayrulewidth}
  |@{\kern-.5\arrayrulewidth}
}
 \hline
  Lengthscale parameter&   Fully nonparametric model& Our proposed marginally constrained model &Parametric model\\ [0.5ex]
 \hline
 \textbf{HMC}&\textbf{-169.33} &\textbf{-164.11} &\textbf{-164.14} \\
 \hline
\end{tabular}
\end{threeparttable}
\end{center}
\end{table*}

\subsection{Real Example 2}


{In our final experiment, we consider modeling 
earthquake data.
Following~\citep{dehghani2020probabilistic}, we are interested in modeling the bivariate distribution of earthquake recurrence time and magnitude, while simultaneously ensuring that the recurrence time follows an exponential distribution~\citep{ferraes2003conditional}.}
We obtain a dataset of 45 observations from~\citet{ferraes2003conditional} (table 1) and run our proposed marginally constrained model with a family constraint. 
Denoting the rate parameter of the exponential family constraint on the recurrence time as $r$, we place a weakly informative gamma prior with both shape and rate parameters to be $0.1$.
To apply our maginally constrained model, we choose a slightly different centering distribution with a same normal-inverse-gamma prior placed on its parameters as that in the first real example described in section~\ref{sec:pm}.
We use the same centering distribution as described in equation~\ref{eq:centering distribution}, where $\theta=(\rho, \mu_2, \sigma_2^2)$, $\phi=r$, $m_x=\frac{1}{r}$ and $s_x^2=\frac{1}{r^2}$. 
The posterior mean density is presented in~\cref{fig:earthquake}.
We see that other than three outliers, the model succeeds in capturing the underlying observation pattern, and that the failure to model the observations arises from the parametric exponential constraint, which effectively robustifies the model against these outliers.
We do not report quantitative performance measures here, essentially, depending on whether or not the outliers are part of the test set, either the fully nonparametric model or our model performs best.
\begin{figure}
  \centerline{
  \includegraphics[width=0.6\textwidth]{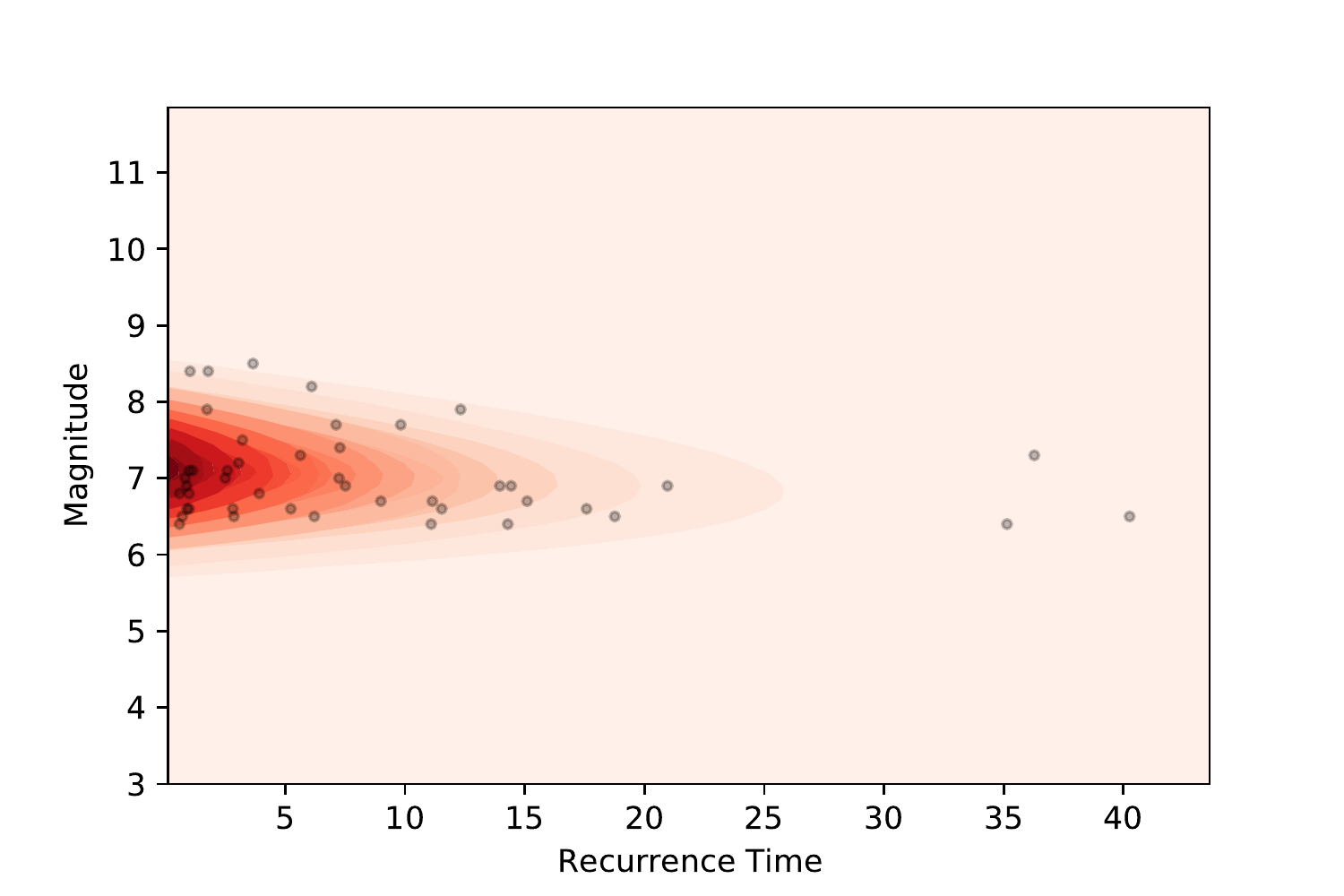}

  }
  \centerline{
  \includegraphics[width=0.6\textwidth]{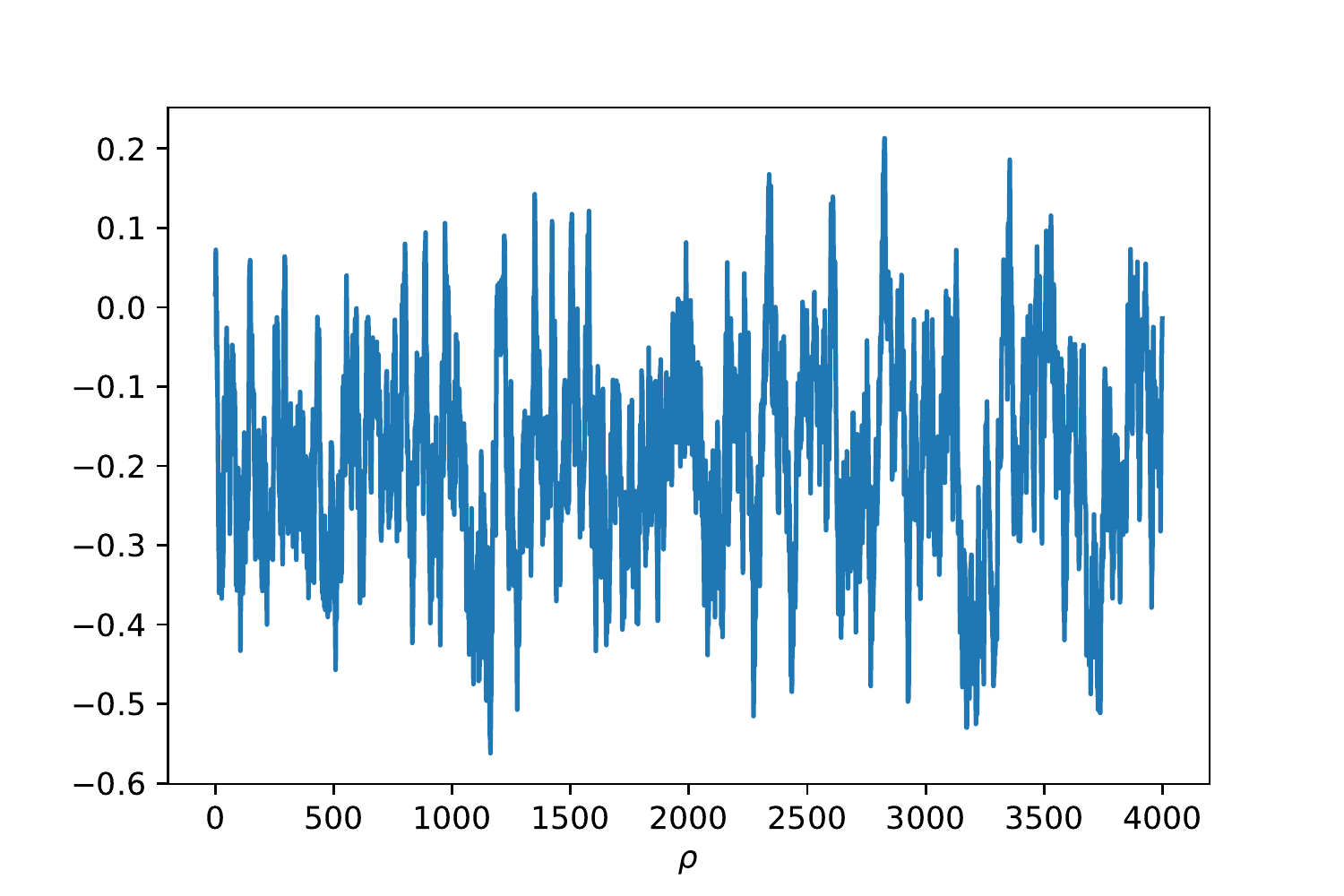}
  \includegraphics[width=0.6\textwidth]{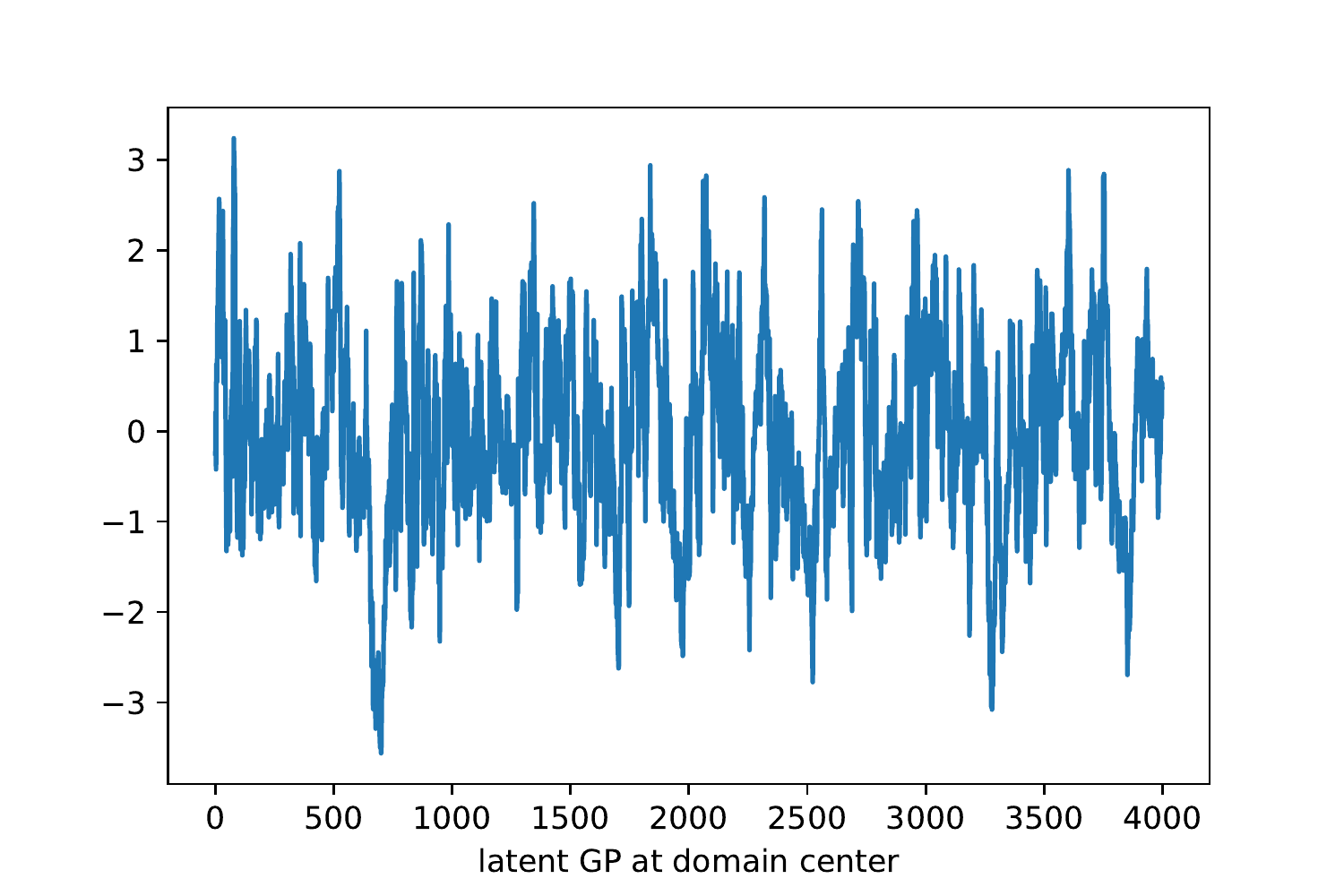}
  }
  \centerline{
    \includegraphics[width=0.6\textwidth]{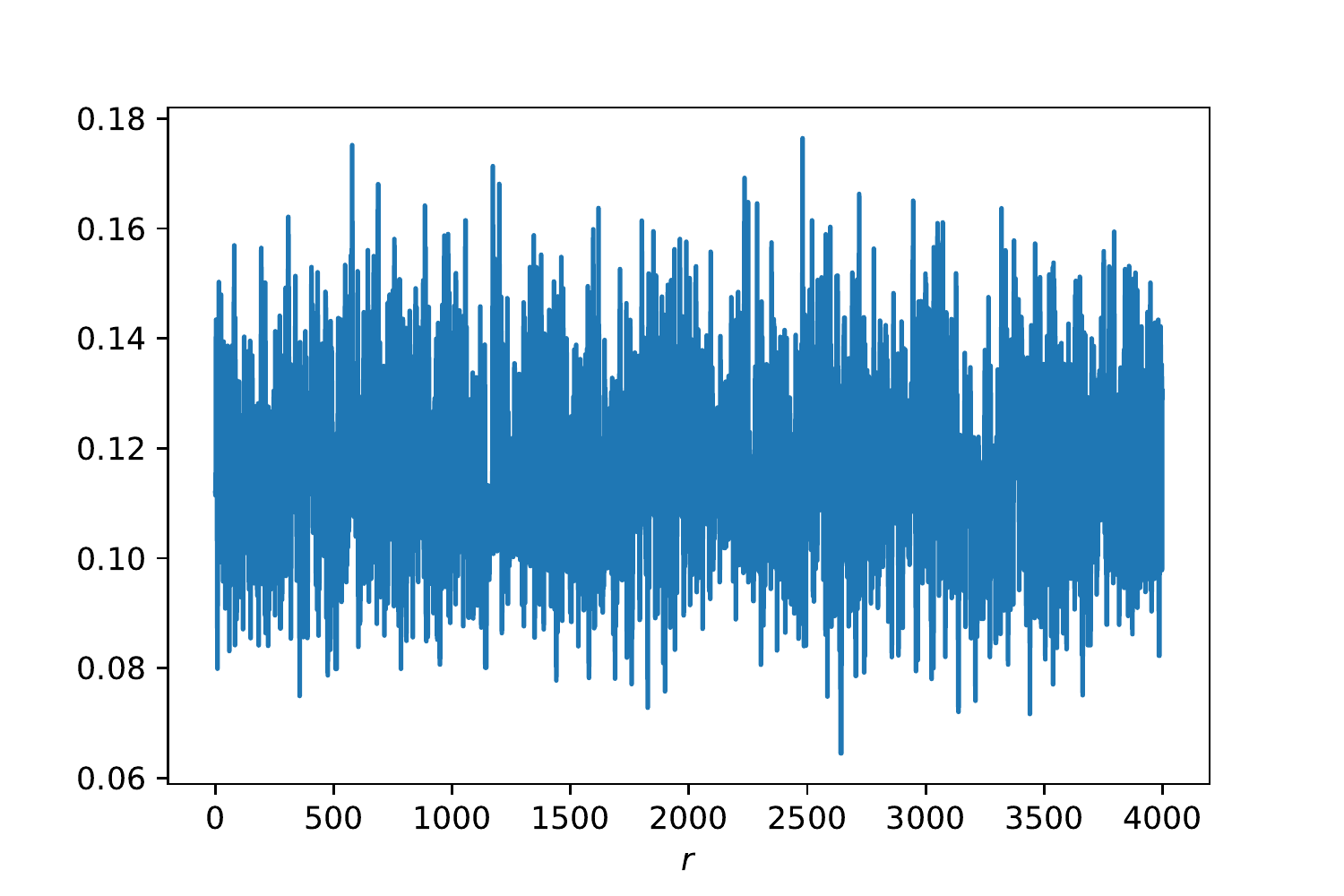}
  \includegraphics[width=0.6\textwidth]{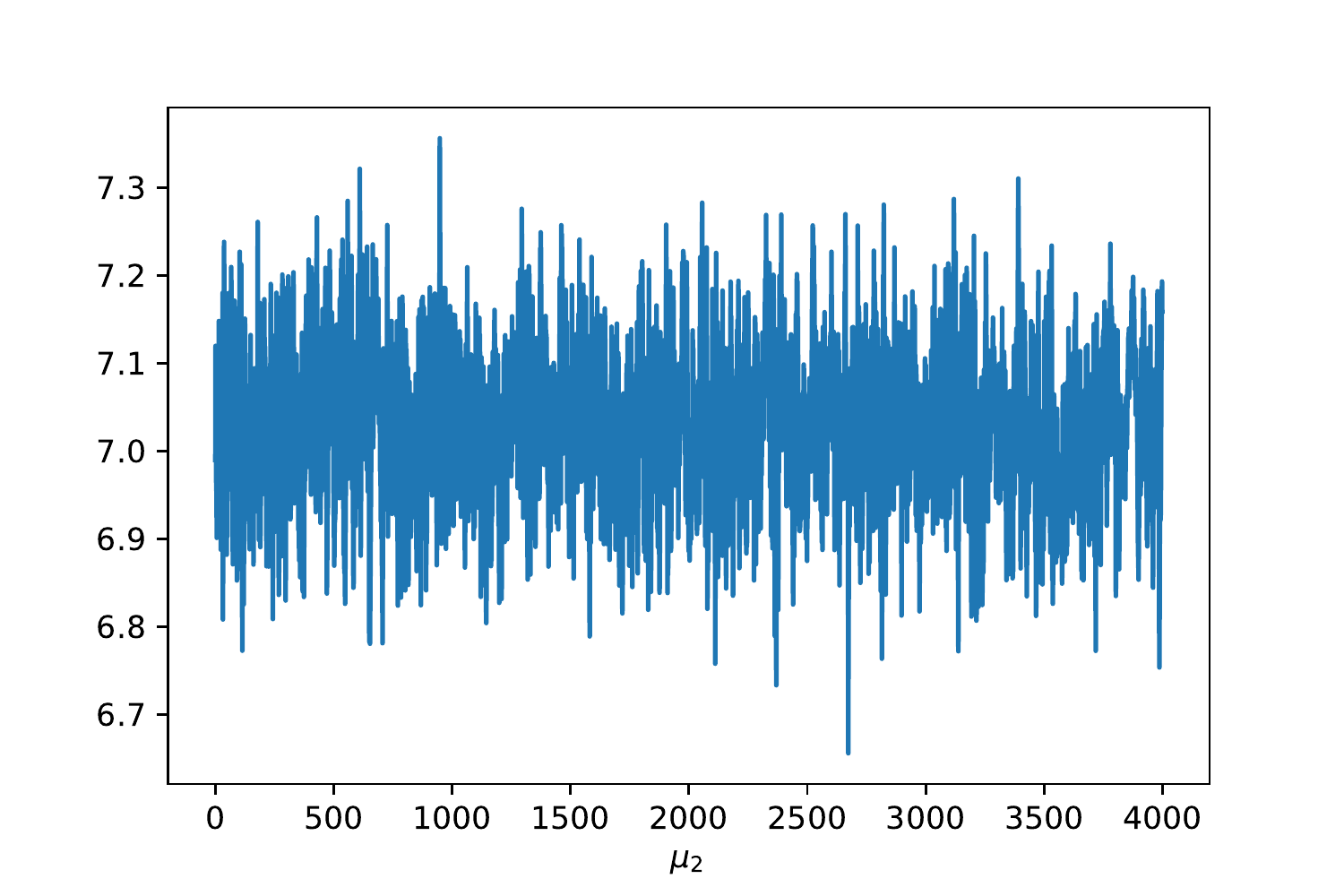}
  }

  \caption{In the top row is the posterior mean density with posterior samples of the lengthscale parameter. 
  The rest are the corresponding traceplots for posterior samples of $\theta$, $\phi$ and the latent GP at the midpoint of the selected gridpoints.}\label{fig:earthquake}
\end{figure}

To assess MCMC mixing, we evaluate the latent GP on a grid of 3540 points and run MCMC for 5000 iterations with a burn-in period of  1000 iterations. In~\cref{table:ess}, we report the minimum ESS, maximum ESS and ESS at the midpoint among all the 3540 grid points. In~\cref{fig:earthquake}, we present the traceplots of posterior samples for $\phi$, $\theta$ and the latent GP at the midpoint.
\begin{table*}
\footnotesize
\begin{center}
\caption{Effective sample size per 4000 iterations for the second real example}
\label{table:ess}
 \begin{tabular}
{
  @{\kern-.5\arrayrulewidth}
  |p{\dimexpr5.5cm-6\tabcolsep-.5\arrayrulewidth}
  |p{\dimexpr5.5cm-6\tabcolsep-.5\arrayrulewidth}
  |p{\dimexpr5.5cm-6\tabcolsep-.5\arrayrulewidth}
  |@{\kern-.5\arrayrulewidth}
}
 \hline
  Minimum ESS & Maximum ESS& ESS at midpoint \\ [0.5ex]
 \hline\hline
  16.76&
242.32&
107.09\\ 
 \hline
\end{tabular}
\end{center}
\end{table*}

\section{Conclusions\label{sec:4}} 

In this work, we propose a nonparametric Bayesian approach for density modeling while enforcing constraints on the marginal distribution of a subset of components. Our approach, closely tied to conditional density modeling introduces a novel constrained Bayesian model based on a transformed Gaussian process, satisfying the marginal constraining distribution exactly and inducing large support prior. For posterior sampling, we devise an exact MCMC algorithm without any approximation/discretization errors, which is additional attraction of  our approach over existing conditional density modeling approaches.

In the present paper, we are only focused on placing one marginal density constraint on a subset of variables.
In some settings, partial prior beliefs are available about different subsets of variables, which requires simultaneously imposing multiple marginal constraints on those subsets of variables. 
Since the dependence between these subsets is unavailable, our approach doesn't extend to it in a straightforward manner. As mentioned at the start of this paper, a more general problem is to constrain some functional of the data distribution. For example, we might have prior information about the mean of the distribution, either in the form of fixed values or prior distribution. In future work, it is of interest to extend our framework to solve these problems. 

There are some other open issues to be considered. First, in this paper we have not discussed sufficient conditions for strong consistency and rates of the convergence of the posterior distribution. \citet{ghosal2000convergence} presents general results on the rates of convergence of the posterior distribution, which can be adapted to our specific transformed Gaussian process prior. Second, we can think about using our proposed transformed Gaussian process prior to solve nonparametric conditional density modeling problems like density regressions. \citet{tokdar2010bayesian} develops a framework for modeling conditional densities and offering dimension reduction of predictors by combining the logistic Gaussian process and subspace projection. A similar future work worth consideration is to connect our proposed transformed Gaussian process prior to subspace projection. More generally, it is of interest to leverage the vast literature on scalability of Gaussian processes to improve the scalability of our proposed model.

\section*{Acknowledgments}
We thank the Editor, Associate Editor and referees, as well as our financial sponsors.

\section*{Appendix}
\section{Appendix A: Weak Posterior Consistency}\label{appendix A}
\subsection{Basics of consistency}
Let $D_1,\dots, D_n$ be i.i.d. with true density $f_0$ belonging to a space of densities $\mathcal{F}$ with weak topology. Let $\Pi$  be a prior on $\mathcal{F}$. For a density $f$, let $P_{f}$ stand for the probability measure corresponding to $f$. Then for any measurable subset $A$ of $\mathcal{F}$, the posterior probability
of $A$ given $D_1, \dots, D_n$ is
$$\Pi(A|D_1, \dots, D_n)=\frac{\int_A \prod_{i=1}^n f(D_i)\Pi(df)}{\int_\mathcal{F}\prod_{i=1}^n f(D_i)\Pi(df)}$$
\begin{definition}[weak neighborhood]
A weak neighborhood of $f_0$ is a set 
$$V_\epsilon(f_0) =\left\{f\in \mathcal{F}: \left|\int\phi f-\int \phi f_0 \right|<\epsilon ,  \text{for any bounded and continuous function}\, \phi\right\}$$.

\end{definition}
\begin{definition}[weak posterior consistency]
  A prior $\Pi$ is said to achieve weak posterior consistency at $f_0$, if $$\Pi(U|X_1,\dots,X_n)\to 1\, \text{almost surely under}\, P_{f_0}$$ for all weak neighborhoods $U$ of $f_0$.
  
\end{definition}
\begin{definition}[KL support]
  $f_0$ is said to be in the KL support of $\Pi$ if $\forall \epsilon>0,\,\Pi\left(K_\epsilon(f_0)\right)>0$ , where $K_\epsilon(f_0) = \{f: \int f_0\log(f_0/f)<\epsilon\}$ is a KL neighborhood of $f_0$ .
\end{definition}
 
 
\subsection{A formal proof of weak posterior convergence for our proposed model}
\begin{lemma}
\label{lm:corresponding theorem 4.1}
For any two functions $\lambda_1$ and $\lambda_2$ on the index set $\mathcal{X}_A\times\mathcal{X}_{A^c}$ and for any $\delta > 0$,  $$\|\lambda_1-\lambda_2\|_\infty<\delta\Rightarrow \left\| \log\frac{f_{\lambda_1}}{f_{\lambda_2}}\right\|_\infty \le 2\delta$$
\end{lemma}

\begin{proof}[\textbf{\upshape Proof:}] 
It follows from $\|\lambda_1-\lambda_2\|_\infty<\delta$ that $\lambda_2\left(X_A,X_{A^c}\right)-\delta < \lambda_1\left(X_A, X_{A^c}\right) <\lambda_2\left(X_A,X_{A^c}\right)+\delta$
for any $X_A, X_{A^c} \in\mathcal{X}_A\times\mathcal{X}_{A^c}$.
Then, from the monotonicity of the sigmoid function $\sigma$, we have that 
$\sigma\left(\lambda_2\left(X_A,X_{A^c}\right)-\delta\right)<\sigma\left(\lambda_1\left(X_A, X_{A^c}\right)\right)<\sigma\left(\lambda_2\left(X_A,X_{A^c}\right)+\delta\right)$.


Next, observe that the sigmoid function satisfies $\sigma(x-\delta) > \exp(-\delta)\sigma(x)$ for any $\delta > 0$. 
%
Combining this with the previous result, we obtain \begin{align*}\exp\left(-\delta\right)\sigma\left(\lambda_2\left(X_A,X_{A^c}\right)\right)<\sigma\left(\lambda_1\left(X_A,X_{A^c}\right)\right)<\exp\left(\delta\right)\sigma\left(\lambda_2\left(X_A,X_{A^c}\right)\right).\end{align*}
It follows that for any density $\pi_0(X_{A^c}|X_A)$, we have 
%
%
\begin{align*} 
 \exp\left(-\delta\right)\int_{\mathcal{X}_{A^c}}\pi_0\left(X_{A^c}|X_A\right)\sigma\left(\lambda_2\left(X_A,X_{A^c}\right)\right)dX_{A^c} & <\int_{\mathcal{X}_{A^c}}\pi_0\left(X_{A^c}|X_A\right)\sigma\left(\lambda_1\left(X_A,X_{A^c}\right)\right)dX_{A^c} \\ 
 & <\exp\left(\delta\right)\int_{\mathcal{X}_{A^c}}\pi_0\left(X_{A^c}|X_A\right)\sigma\left(\lambda_2\left(X_A,X_{A^c}\right)\right)dX_{A^c}
 \end{align*}
From the above two inequalities, it follows that
\begin{align*}
    \exp\left(-2\delta\right)\cfrac{\sigma\left(\lambda_2\left(X_A,X_{A^c}\right)\right)}{\int_{\mathcal{X}_{A^c}}\pi_0\left(X_{A^c}|X_A\right)\sigma\left(\lambda_2\left(X_A,X_{A^c}\right)\right)dX_{A^c}} 
   & < \cfrac{\sigma\left(\lambda_1\left(X_A,X_{A^c}\right)\right)}{\int_{\mathcal{X}_{A^c}}\pi_0\left(X_{A^c}|X_A\right)\sigma\left(\lambda_1\left(X_A,X_{A^c}\right)\right)dX_{A^c}} \\
& <\exp\left(2\delta\right)\cfrac{\sigma\left(\lambda_2\left(X_A,X_{A^c}\right)\right)}{\int_{\mathcal{X}_{A^c}}\pi_0\left(X_{A^c}|X_A\right)\sigma\left(\lambda_2\left(X_A,X_{A^c}\right)\right)dX_{A^c}},
\end{align*}
so that for all $(X_A,X_{A^c})$,
\begin{align*}
    \exp\left(-2\delta\right) < \frac{f_{\lambda_1}\left(X_A,X_{A^c}\right)}{f_{\lambda_2}\left(X_A,X_{A^c}\right)} <  \exp\left(2\delta\right) 
\end{align*}
The result then follows.
\end{proof}
\abc*
\begin{proof}[\textbf{\upshape Proof:}] 

By definition, any density that belongs to $\mathcal{F}$ takes the form $p_A(X_A)\cdot q(X_{A^c}|X_A)$.  
For any $\epsilon>0$, set $\delta=\cfrac{\epsilon}{4}$, and choose $\xi>0$ such that $\log(1+\xi)<\frac{\epsilon}{2}$.
%
%
Define a strictly positive density $q_0$ as 
$$
q_0\left(X_{A^c}|X_A\right)=\cfrac{q\left(X_{A^c}|X_A\right)+\xi}{1+\xi} \qquad\forall X_A\in\mathcal{X}_A, X_{A^c}\in \mathcal{X}_{A^c}.
$$
Consider the ratio $\cfrac{q_0\left(X_{A^c}|X_A\right)}{\pi_0\left(X_{A^c}|X_A\right)}$.
As both $q_0$ and $\pi_0$ are continuous functions on $\mathcal{X}_A\times\mathcal{X}_{A^c}$, and as $\pi_0$ does not vanish, this is also a continuous function on $\mathcal{X}_A\times\mathcal{X}_{A^c}$. 
Due to the compactness of $\mathcal{X}_A\times\mathcal{X}_{A^c}$, we additionally have that $M:=\sup_{\left\{X_A\in \mathcal{X}_A, X_{A^c} \in \mathcal{X}_{A^c}\right\}} \cfrac{q_0\left(X_{A^c}| X_A\right)}{\pi_0\left(X_{A^c}|X_A\right)}<\infty$. 
Define $\lambda_0(X_A,X_{A^c})=\sigma^{-1}(\frac{q_0(X_{A^c}| X_A)}{M\pi_0(X_{A^c}|X_A)})$.
The definition of $q_0$ and $M$, and the fact that $\pi_0$ does not vanish ensures that $\frac{q_0(X_{A^c}| X_A)}{M\pi_0(X_{A^c}|X_A)} \in (0,1)$, and thus lies in the domain of $\sigma^{-1}(\cdot)$.
Recalling the mapping $f_\lambda$ is defined in~\cref{eq:map}, it is easy to see that the function $\lambda_0$ satisfies
$f_{\lambda_0}(X_A,X_{A^c}) = p_A(X_A) q_0(X_{A'}|X_A)$. 



From the assumptions of the theorem, it follows from~\cref{lm:gaussian process property1} that 
$$P(\lambda:\sup_{\left\{X_A\in \mathcal{X}_A, X_{A^c} \in \mathcal{X}_{A^c}\right\}}|\lambda(X_A, X_{A^c})-\lambda_0(X_A,X_{A^c})|<\delta)>0.$$
From~\cref{lm:corresponding theorem 4.1}, we then obtain
\begin{align*}
&\Pi\left(f_\lambda:\left\|\log\cfrac{f_{\lambda_0}(X_A, X_{A^c})}{f_\lambda(X_A, X_{A^c})}\right\|_\infty<\cfrac{\epsilon}{2}\right)\geq P(\lambda:\sup_{\left\{X_A \in \mathcal{X}_A,X_{A^c}\in \mathcal{X}_{A^c}\right\}}|\lambda\left(X_A, X_{A^c}\right)-\lambda_0\left(X_A,X_{A^c}\right)|<\delta)>0.
\end{align*}
Now, recognizing that $\log\cfrac{q\left(X_{A^c}|X_A\right)}{q\left(X_{A^c}|X_A\right)+\xi} < 0$, we have:
\begin{align*}
\text{KL}\,&\left(p_A\left(X_A\right)\cdot q\left(X_{A^c}|X_A\right),f_\lambda\left(X_A, X_{A^c}\right)\right)=\int_{\mathcal{X}_A}\int_{\mathcal{X}_{A^c}}p_A\left(X_A\right)q\left(X_{A^c}|X_A\right)\log\cfrac{p_A\left(X_A\right)q\left(X_{A^c}|X_A\right)}{f_\lambda\left(X_A, X_{A^c}\right)}dX_{A^c}dX_A\\
&=\int_{\mathcal{X}_A}\int_{\mathcal{X}_{A^c}}P_A\left(X_A\right)q\left(X_{A^c}|X_A\right)\log\cfrac{q\left(X_{A^c}|X_A\right)}{q_0\left(X_{A^c}|X_A\right)}dX_{A^c}dX_A+\int_{\mathcal{X}_A}\int_{\mathcal{X}_{A^c}}p_A\left(X_A\right)q\left(X_{A^c}|X_A\right)\log\cfrac{p_A\left(X_A\right)q_0\left(X_{A^c}|X_A\right)}{f_\lambda\left(X_A, X_{A^c}\right)}dX_{A^c}dX_A\\
&=\int_{\mathcal{X}_A}\int_{\mathcal{X}_{A^c}}p_A\left(X_A\right)q\left(X_{A^c}|X_A\right)\log\cfrac{q\left(X_{A^c}|X_A\right)}{q\left(X_{A^c}|X_A\right)+\xi}dX_{A^c}dX_A+\log\left(1+\xi\right)\\
&\quad +\int_{\mathcal{X}_A}\int_{\mathcal{X}_{A^c}}p_A\left(X_A\right)q\left(X_{A^c}|X_A\right)\log\cfrac{p_A\left(X_A\right)q_0\left(X_{A^c}|X_A\right)}{f_\lambda\left(X_A, X_{A^c}\right)}dX_{A^c}dX_A\\
&< \frac{\epsilon}{2}+\left\|\log\cfrac{f_{\lambda_0}\left(X_A, X_{A^c}\right)}{f_\lambda\left(X_A, X_{A^c}\right)}\right\|_\infty.
\end{align*}
%
It follows that
$\Pi(f_\lambda:\text{KL}(p_A\cdot q,f_\lambda)<\epsilon)\geq \Pi\left(f_\lambda:\left\|\log\cfrac{f_{\lambda_0}(X_1, X_2)}{f_\lambda(X_1, X_2)}\right\|_\infty<\cfrac{\epsilon}{2}\right)>0$, completing the proof.
%
\end{proof}

\bibliographystyle{unsrtnat}  
\bibliography{trial}

\begin{thebibliography}{31}
\providecommand{\natexlab}[1]{#1}
\providecommand{\url}[1]{\texttt{#1}}
\expandafter\ifx\csname urlstyle\endcsname\relax
  \providecommand{\doi}[1]{doi: #1}\else
  \providecommand{\doi}{doi: \begingroup \urlstyle{rm}\Url}\fi

\bibitem[Ferguson(1973)]{ferguson}
Thomas~S Ferguson.
\newblock {A Bayesian analysis of some nonparametric problems}.
\newblock \emph{The Annals of Statistics}, pages 209--230, 1973.

\bibitem[Rasmussen and Williams(2006)]{rasmussen}
CE. Rasmussen and CKI. Williams.
\newblock \emph{{Gaussian Processes for Machine Learning}}.
\newblock Adaptive Computation and Machine Learning. MIT Press, Cambridge, MA,
  USA, January 2006.

\bibitem[Dunson(2010)]{dunson}
David~B Dunson.
\newblock {Nonparametric Bayes applications to biostatistics}.
\newblock \emph{Bayesian Nonparametrics}, 28:\penalty0 223--273, 2010.

\bibitem[Teh and Jordan(2009)]{teh}
Yee~Whye Teh and Michael~I Jordan.
\newblock {Hierarchical Bayesian nonparametric models with applications}.
\newblock \emph{Bayesian Nonparametrics}, 28\penalty0 (158):\penalty0 42, 2009.

\bibitem[Sudderth and Jordan(2008)]{sudderth}
Erik Sudderth and Michael Jordan.
\newblock {Shared segmentation of natural scenes using dependent Pitman-Yor
  processes}.
\newblock \emph{Advances in Neural Information Processing Systems},
  21:\penalty0 1585--1592, 2008.

\bibitem[Kessler et~al.(2015)Kessler, Hoff, and Dunson]{kessler}
David~C Kessler, Peter~D Hoff, and David~B Dunson.
\newblock {Marginally specified priors for non-parametric Bayesian estimation}.
\newblock \emph{Journal of the Royal Statistical Society, Series B (Statistical
  methodology)}, 77\penalty0 (1):\penalty0 35, 2015.

\bibitem[Schifeling and Reiter(2016)]{schifeling}
Tracy~A Schifeling and Jerome~P Reiter.
\newblock {Incorporating marginal prior information in latent class models}.
\newblock \emph{Bayesian Analysis}, 11\penalty0 (2):\penalty0 499--518, 2016.

\bibitem[Dai et~al.(2022)Dai, Yang, Xue, Schuurmans, and Dai]{dai2022marginal}
Hanjun Dai, Mengjiao Yang, Yuan Xue, Dale Schuurmans, and Bo~Dai.
\newblock {Marginal distribution adaptation for discrete sets via
  module-oriented divergence minimization}.
\newblock In \emph{International Conference on Machine Learning}, pages
  4605--4617. PMLR, 2022.

\bibitem[Dunson and Park(2008)]{dunson2008kernel}
David~B Dunson and Ju~Hyun Park.
\newblock {Kernel stick-breaking processes}.
\newblock \emph{Biometrika}, 95\penalty0 (2):\penalty0 307--323, 2008.

\bibitem[Chung and Dunson(2009)]{chung2009nonparametric}
Yeonseung Chung and David~B Dunson.
\newblock {Nonparametric Bayes conditional distribution modeling with variable
  selection}.
\newblock \emph{Journal of the American Statistical Association}, 104\penalty0
  (488):\penalty0 1646--1660, 2009.

\bibitem[Pati et~al.(2013)Pati, Dunson, and Tokdar]{pati2013posterior}
Debdeep Pati, David~B Dunson, and Surya~T Tokdar.
\newblock {Posterior consistency in conditional distribution estimation}.
\newblock \emph{Journal of Multivariate Analysis}, 116:\penalty0 456--472,
  2013.

\bibitem[Ghosh et~al.(2010)Ghosh, Tokdar, and Zhu]{ghosh2010bayesian}
Jayanta~K Ghosh, Surya~T Tokdar, and Yu~M Zhu.
\newblock {Bayesian density regression with logistic Gaussian process and
  subspace projection}.
\newblock \emph{Bayesian Analysis}, 5\penalty0 (2):\penalty0 319--344, 2010.

\bibitem[Tokdar(2011)]{tokdar2011dimension}
Surya~T Tokdar.
\newblock {Dimension adaptability of Gaussian process models with variable
  selection and projection}.
\newblock \emph{arXiv preprint arXiv:1112.0716}, 2011.

\bibitem[Adams et~al.(2008)Adams, Murray, and MacKay]{adams}
Ryan~P Adams, Iain Murray, and David MacKay.
\newblock {The Gaussian process density sampler}.
\newblock \emph{Advances in Neural Information Processing Systems}, 21, 2008.

\bibitem[Lenk(1988)]{lenk1988logistic}
Peter~J Lenk.
\newblock {The logistic normal distribution for Bayesian, nonparametric,
  predictive densities}.
\newblock \emph{Journal of the American Statistical Association}, 83\penalty0
  (402):\penalty0 509--516, 1988.

\bibitem[Lenk(1991)]{lenk1991towards}
Peter~J Lenk.
\newblock {Towards a practicable Bayesian nonparametric density estimator}.
\newblock \emph{Biometrika}, 78\penalty0 (3):\penalty0 531--543, 1991.

\bibitem[Leonard(1978)]{leonard1978density}
Tom Leonard.
\newblock {Density estimation, stochastic processes and prior information}.
\newblock \emph{Journal of the Royal Statistical Society, Series B (Statistical
  Methodology)}, 40\penalty0 (2):\penalty0 113--132, 1978.

\bibitem[Tokdar(2007)]{tokdar2007towards}
Surya~T Tokdar.
\newblock {Towards a faster implementation of density estimation with logistic
  Gaussian process priors}.
\newblock \emph{Journal of Computational and Graphical Statistics}, 16\penalty0
  (3):\penalty0 633--655, 2007.

\bibitem[Rao et~al.(2016)Rao, Lin, and Dunson]{rao}
Vinayak Rao, Lizhen Lin, and David~B Dunson.
\newblock {Data augmentation for models based on rejection sampling}.
\newblock \emph{Biometrika}, 103\penalty0 (2):\penalty0 319--335, 2016.

\bibitem[Choudhuri et~al.(2007)Choudhuri, Ghosal, and
  Roy]{choudhuri2007nonparametric}
Nidhan Choudhuri, Subhashis Ghosal, and Anindya Roy.
\newblock {Nonparametric binary regression using a Gaussian process prior}.
\newblock \emph{Statistical Methodology}, 4\penalty0 (2):\penalty0 227--243,
  2007.

\bibitem[Ghosal and Roy(2006)]{ghosal2006posterior}
Subhashis Ghosal and Anindya Roy.
\newblock {Posterior consistency of Gaussian process prior for nonparametric
  binary regression}.
\newblock \emph{The Annals of Statistics}, 34\penalty0 (5):\penalty0
  2413--2429, 2006.

\bibitem[Tokdar and Ghosh(2007)]{tokdar2007posterior}
Surya~T Tokdar and Jayanta~K Ghosh.
\newblock {Posterior consistency of logistic Gaussian process priors in density
  estimation}.
\newblock \emph{Journal of Statistical Planning and Inference}, 137\penalty0
  (1):\penalty0 34--42, 2007.

\bibitem[Schwartz(1965)]{schwartz1965bayes}
Lorraine Schwartz.
\newblock {On Bayes procedures}.
\newblock \emph{Zeitschrift f{\"u}r Wahrscheinlichkeitstheorie und verwandte
  Gebiete}, 4\penalty0 (1):\penalty0 10--26, 1965.

\bibitem[Murray et~al.(2006)Murray, Ghahramani, and MacKay]{murray2012mcmc}
I~Murray, Z~Ghahramani, and D~MacKay.
\newblock {MCMC for doubly-intractable distributions}.
\newblock \emph{Uncertainty in Artificial Intelligence}, 22, 2006.

\bibitem[Murray et~al.(2010)Murray, Adams, and MacKay]{murray2010elliptical}
Iain Murray, Ryan Adams, and David MacKay.
\newblock {Elliptical slice sampling}.
\newblock In \emph{Proceedings of the Thirteenth International Conference on
  Artificial Intelligence and Statistics}, pages 541--548. JMLR Workshop and
  Conference Proceedings, 2010.

\bibitem[Neal(2011)]{neal2011mcmc}
Radford~M Neal.
\newblock {MCMC using Hamiltonian dynamics}.
\newblock \emph{Handbook of Markov Chain Monte Carlo}, 2\penalty0
  (11):\penalty0 2, 2011.

\bibitem[Ott(1990)]{ott1990physical}
Wayne~R Ott.
\newblock {A physical explanation of the lognormality of pollutant
  concentrations}.
\newblock \emph{Journal of the Air \& Waste Management Association},
  40\penalty0 (10):\penalty0 1378--1383, 1990.

\bibitem[Dehghani and Fadaee(2020)]{dehghani2020probabilistic}
Hamzeh Dehghani and Mohammad~Javad Fadaee.
\newblock {Probabilistic prediction of earthquake by bivariate distribution}.
\newblock \emph{Asian Journal of Civil Engineering}, 21:\penalty0 977--983,
  2020.

\bibitem[Ferraes(2003)]{ferraes2003conditional}
Sergio~G Ferraes.
\newblock {The conditional probability of earthquake occurrence and the next
  large earthquake in Tokyo, Japan}.
\newblock \emph{Journal of Seismology}, 7\penalty0 (2):\penalty0 145--153,
  2003.

\bibitem[Ghosal et~al.(2000)Ghosal, Ghosh, and Van
  Der~Vaart]{ghosal2000convergence}
Subhashis Ghosal, Jayanta~K Ghosh, and Aad~W Van Der~Vaart.
\newblock {Convergence rates of posterior distributions}.
\newblock \emph{Annals of Statistics}, pages 500--531, 2000.

\bibitem[Tokdar et~al.(2010)Tokdar, Zhu, and Ghosh]{tokdar2010bayesian}
Surya~T Tokdar, Yu~M Zhu, and Jayanta~K Ghosh.
\newblock {Bayesian density regression with logistic Gaussian process and
  subspace projection}.
\newblock \emph{Bayesian Analysis}, 5\penalty0 (2):\penalty0 319--344, 2010.

\end{thebibliography}

\end{document}